\documentclass[11pt]{article}
\usepackage{latexsym}
\usepackage{amsmath,amssymb,amsthm}
\usepackage{epsfig}
\usepackage{fancyhdr}
\usepackage[square, numbers]{natbib}
\usepackage[utf8]{inputenc}
\usepackage[boxed]{algorithm2e}
\usepackage{algorithmic}
\usepackage{url}
\usepackage{hyperref}
\usepackage[english]{babel}
\usepackage{geometry}
\usepackage{comment}
\usepackage{cleveref}
\usepackage{color}
\usepackage{bm}
\usepackage{fullpage}
\usepackage{authblk}
\usepackage{enumitem}

\renewcommand{\Re}{\mathbb{R}}
\newcommand{\R}{\mathbb{R}}
\newcommand{\Z}{\mathbb{Z}}
\newcommand{\dist}{D}

\newcommand{\eps}{\varepsilon}

\DeclareMathOperator{\poly}{poly}

\newcommand{\ceil}[1]{\lceil #1 \rceil}

\newcommand{\sI}{\mathcal{I}}

\DeclareMathOperator*{\argmin}{arg\,min}
\DeclareMathOperator{\polylog}{polylog}
\DeclareMathOperator{\Var}{Var}
\DeclareMathOperator{\E}{\mathbb{E}}

\theoremstyle{plain}
\newtheorem{theorem}{Theorem}[section]

\newtheorem{lemma}[theorem]{Lemma}
\newtheorem{corollary}[theorem]{Corollary}

\newtheorem{observation}[theorem]{Observation}

\usepackage{blindtext}
\usepackage{titling}
\usepackage{array}
\preauthor{\begin{center}
    \large \lineskip .75em%
        \begin{tabular}[t]{>{\centering\arraybackslash}p{.45\textwidth}}}
        \postauthor{\end{tabular}\par\end{center}}
\makeatletter
\renewcommand\and{
  \end{tabular}%
  \hfill
  \begin{tabular}[t]{>{\centering\arraybackslash}p{.45\textwidth}}}
  \makeatother

\author{
  Alexandr Andoni\authorcr
  \small{Columbia University}\authorcr
  \small{\texttt{andoni@cs.columbia.edu}}
  \and
  Collin Burns\authorcr
  \small{Columbia University}\authorcr
  \small{\texttt{collin.burns@columbia.edu}}
  \and
  Yi Li\authorcr
  \small{Nanyang Technological University}\authorcr
  \small{\texttt{yili@ntu.edu.sg}}
  \and
  Sepideh Mahabadi\authorcr
  \small{Toyota Technological Institute at Chicago}\authorcr
  \small{\texttt{mahabadi@ttic.edu}}
  \and
  David P. Woodruff\authorcr
  \small{Carnegie Mellon University}\authorcr
  \small{\texttt{dwoodruf@cs.cmu.edu}}
}

\title{Streaming Complexity of SVMs}

\begin{document}

\date{}
\maketitle
\thispagestyle{empty}

\begin{abstract}
We study the space complexity of solving the bias-regularized SVM problem
in the streaming model. In particular, given a data set
$(x_i,y_i)\in \R^d\times \{-1,+1\}$, the objective function is
$F_\lambda(\theta,b) = \tfrac{\lambda}{2}\|(\theta,b)\|_2^2 + \tfrac{1}{n}\sum_{i=1}^n
\max\{0,1-y_i(\theta^Tx_i+b)\}$
and the goal is to find the parameters that (approximately) minimize this objective. This is a classic supervised learning problem that has drawn lots
of attention, including for developing fast algorithms for solving the
problem approximately: i.e., for finding $(\theta,b)$ such that $F_\lambda(\theta,b)\le
\min_{(\theta',b')} F_\lambda(\theta',b')+\eps$.

One of the most widely used algorithms for approximately optimizing the SVM
objective is Stochastic Gradient Descent (SGD), which requires only
$O(\tfrac{1}{\lambda\eps})$ random samples, and which immediately yields a streaming algorithm that uses $O(\tfrac{d}{\lambda\eps})$ space.
For related problems, better streaming algorithms are only known for {\em smooth}
functions, unlike the SVM objective that we focus on in this work.

We initiate an investigation of the space complexity for
both finding an approximate optimum of this objective, and for the related ``point
estimation'' problem of sketching the data set to evaluate the function value $F_\lambda$
on any query $(\theta, b)$. We show that, for both problems, for
dimensions $d=1,2$, one can obtain streaming algorithms with space polynomially
smaller than $\tfrac{1}{\lambda\eps}$, which is the 
complexity of SGD for strongly convex functions like the bias-regularized SVM \cite{shalev-shwartz2007pegasos}, and which is known to be tight in general, even for $d=1$ \cite{agarwal2009samplecomplexity}. We also prove polynomial lower bounds for both point
estimation and optimization. In particular, for point estimation we obtain a
tight bound of $\Theta(1/\sqrt{\eps})$ for $d=1$ and a nearly tight lower bound of $\widetilde{\Omega}(d/{\eps}^2)$ for $d = \Omega( \log(1/\epsilon))$. Finally, for optimization, we prove a $\Omega(1/\sqrt{\epsilon})$ lower bound for $d = \Omega( \log(1/\epsilon))$, and show similar bounds when $d$ is constant.

\end{abstract}
\newpage 

\section{Introduction}

The Support Vector Machine (SVM) optimization problem is a classic supervised learning problem with a rich and extensive literature. In this work, we consider the SVM problem in the {\em space-constrained}
setting.  Specifically, we focus on the bias-regularized SVM\footnote{In the standard SVM formulation, the bias is not regularized, but the bias-regularized version is common both in theoretical work and in practice. See, for example, \cite{shalev-shwartz2007pegasos}.}. For $n$ labelled data points $(x_i,y_i)\in
\R^d\times \{-1, +1\}$, with $\|x_i\| \leq 1$, and $(\theta, b)\in
\R^{d} \times \R $ the unknown model parameters, the SVM objective function is
defined as:
\begin{equation}\label{eqn:F_lambda}
  F_\lambda(\theta, b) := 
\frac{\lambda}{2}\|(\theta, b)\|_2^2 + \frac{1}{n}\sum_{i=1}^n \max\{0,1-y_i(\theta^T x_i + b)\},
\end{equation}
where $\lambda$ is the regularization parameter. The SVM optimization problem is then to minimize the objective: 
\begin{equation}
    (\theta^*,b^*):=\argmin_{\theta, b} F_\lambda(\theta, b).
\end{equation}

This problem is of both theoretical and practical interest, and has
received lots of attention in the machine learning community. One of
the main lines of work on SVMs has focused on trying to find
approximately optimal solutions quickly (see, e.g.,
\cite{shalev-shwartz2007pegasos}, \cite{backurs2017fine}, \cite{shalev2014understanding}, and the references therein). Most notably, using a
variant of stochastic gradient descent (SGD), one can compute a
solution $(\hat \theta,\hat b)$ which is at most $\eps$ away
from the optimal $F_\lambda(\theta^*,b^*)$ in
$O(\tfrac{1}{\lambda\eps})$ SGD steps, each using a single randomly
sampled data point $(x_i,y_i)$ \cite{shalev-shwartz2007pegasos}. 

However, in many applications of SVMs, the number of data points is sufficiently large
that even storing all of the data may be prohibitively expensive.
In this case, it may be desirable to store a smaller summary of the data that is alone sufficient for optimizing the SVM objective within a desired error tolerance.

This goal has been studied for related smooth objective functions,
but not for non-smooth objectives like SVMs. For example,
\cite{pass-glm} focuses on this problem for the more general setting of
Generalized Linear Models (GLMs), where $F$ is defined as
$F=\tfrac{1}{n}\sum_{i=1}^n \phi(y_i, \theta^T x_i)$ for an arbitrary function
$\phi$. This includes the SVM objective. The authors of
\cite{pass-glm} show that if $\phi$ is well-approximable by a
low-degree polynomial, then one can stream through the data points
while keeping a small sketch of the data that is sufficient for minimizing
the objective. However, the space complexity depends
exponentially on the degree of the approximating polynomial, so it is
only feasible for relatively smooth functions $\phi$.

In this work we study SVMs, the most common non-smooth GLM, and focus on space complexity in the streaming setting. 
In addition to optimization, we also focus on the
problem of \textit{point estimation}: sketching the data points so that, given $(\theta, b)$, we can output a value within $\eps$ of $F_\lambda(\theta,b)$. While we use this as an intermediate step for achieving improved optimization upper bounds in the low-dimensional setting, this problem is also of
independent interest. It occurs, for example, in estimating the GLM
posterior distribution: the distribution of the parameters $(\theta, b)$ given the observed data and some prior distribution over $(\theta, b)$ \cite{pass-glm}. 

\subsection{Our Results}

Our results are for the two considered problems, and include both upper and
lower bounds. We present the results for the point
estimation problem independently. Then we present our results for
optimization, the upper bounds for which rely on our point estimation results.

\subparagraph*{Point estimation.}
First, we show that one can obtain a {\em multiplicative}
$(1+\eps)$-factor approximation for $d=1$ that uses $\tilde O(1 / \eps^2)$ space.
However, we then show that it is not possible to get a multiplicative approximation algorithm for dimension $d>1$ that uses space sub-linear in $n$.
Given this, we otherwise focus on the space complexity of {\em additive} $\pm \eps$ approximation streaming algorithms. For
$d=1$, we obtain space $O(1/\sqrt{\eps})$, and for $d=2$, we obtain
space $O(\eps^{-4/5})$.
We complement our algorithms with a lower
bound of $\Omega(\eps^{-(d+1)/(d+3)})$ for dimension $d$, for
any sketching algorithm, which goes up to $\eps^{-1}$ as we increase $d$.
Note that, for $d=1$, our bounds are tight. For $d=2$, our lower bound
translates to $\Omega(\eps^{-3/5})$.

We also prove a lower bound of $\Omega(d/(\eps^2\polylog(1/\eps)))$ for
$d = \Omega(\log(1/\eps))$, which is tight up to $\polylog(1/\eps)$ factors. 
Together with the $O(1/\lambda\eps)$ upper bound achieved from SGD, this shows that there is a \emph{strict gap} between point estimation and optimization. This is also the case for linear regression in the streaming model: while getting a multiplicative $1\pm \eps$ point estimation approximation for linear regression requires space $\tilde \Theta(\eps^{-2})$ \cite{jayram2013jll}, the (multiplicative) optimization problem requires only $\tilde \Theta(\eps^{-1})$ space \cite{clarkson2009la_streaming}.





\subparagraph*{Optimization.}
First, using standard net arguments, we show how to use our point estimation results to approximately minimize the SVM objective. 

\begin{theorem}
    Suppose there is a streaming algorithm that, after seeing data $\{(x_i,y_i)\}_{i=1}^n$, where $\|x_i\| \leq 1$, can produce a sketch of size $s$ that, given any $(\theta,b)$ such that $\|(\theta, b)\| \leq \sqrt{2/\lambda}$, is able to output $\hat F(\theta, b)$ such that $|\hat F(\theta,b)- F(\theta, b)|\leq \eps$ with probability at least $0.9$. Then there is also a streaming algorithm that, under the same input, will output $(\hat \theta, \hat b)$ with $|F_\lambda(\hat\theta,\hat b)-F_\lambda(\theta^*_\lambda,b^*_\lambda)|\le \eps$ with probability at least $0.9$, while using space $O(s \cdot d\log d/(\lambda \eps))$.
\end{theorem}

Together with our point estimation results, we immediately obtain
streaming algorithms that, for dimensions one and two, obtain space
polynomially smaller than $\tfrac{1}{\lambda\eps}$, which is the 
complexity of SGD for strongly convex functions like the bias-regularized SVM \cite{shalev-shwartz2007pegasos}, and which is tight in general \cite{agarwal2009samplecomplexity}. 

We also prove space lower bounds for optimization. First, we consider the high-dimensional case when $\lambda = 10^{-4}$ is a constant:

\begin{theorem}
Let $\lambda = 10^{-4}$ and $d = O(\log n)$ for $n = \Theta(\eps^{-1/2})$. Suppose there exists a sketch such that, given a stream of inputs $\{(x_i, y_i)\}_{i=1}^n$, outputs some $(\hat{\theta}, \hat{b})$ with probability at least $0.9$ such that 
    $F_\lambda(\hat{\theta}, \hat{b}) \leq F_\lambda(\theta^*, b^*) + \eps$.
Then such a sketch requires $\Omega(\eps^{-1/2})$ space.
\end{theorem}

We can adapt the lower bound to the low dimensional setting if we let $\lambda = \poly(1/n)$. Specifically, we show a somewhat weaker lower bound of $\Omega(\eps^{-1/4})$ for $d=2$ as long as $\lambda = \Theta(1/n^2)$. Moreover, we are able to show the same $\Omega(\eps^{-1/2})$ bound for any $d \geq 3$ as long as $\lambda = \Theta(1/n)$. Note that $\lambda = 1/n$ is a reasonable setting that is often used in practice. 


\subsection{Related Work}

\subparagraph*{Stochastic optimization methods.}
SGD for strongly convex functions, which includes the bias-regularized SVM, has a sample complexity $O(\frac{1}{\lambda \eps})$. Consequently, in the streaming setting we can simply maintain $O(\frac{1}{\lambda \eps})$ random elements from the stream, then at the end run SGD, yielding a space complexity of $O(\frac{d}{\lambda \eps})$. Moreover, this is tight for general Lipschitz, strongly convex functions; there is an $\Omega(\frac{1}{\eps})$ SGD sample complexity lower bound for this function class~\cite{agarwal2009samplecomplexity}. 

Some stochastic optimization methods like Stochastic Average Gradient (SAG) \cite{schmidt2017sga} achieve \textit{linear} convergence, meaning that after $T$ iterations $F_\lambda(\theta_T) - F_\lambda(\theta^*) \leq O(\rho^T)$ for some $\rho < 1$. However, in this case, $\rho \approx 1-1/n$, so we would need $T > n$ iterations, which is worse than $\frac{1}{\lambda \eps}$ when $n$ is sufficiently large. 

Similarly, some stochastic optimization methods like Katyusha \cite{allen-zhu2017katyusha} achieve a sample complexity that has a dependence on $\eps$ like $1/\sqrt{\eps}$. However, the sample complexity for such methods also has a sample complexity that scales linearly with $n$, again making them worse than $\frac{1}{\lambda \eps}$ for the regime we care about.

Finally, note that if the elements of the stream are drawn IID from some distribution, and the size of the stream, $n$, is sufficiently large, then we can
simply run online gradient descent (OGD) and use $O(d)$
space. However, we focus on the general setting where we cannot
make any distributional assumptions.

\subparagraph*{Core-set and streaming algorithms for SVMs.}
Tsang et al.\ focuses on trying to speed up SVM optimization via \textit{core-sets} \cite{tsang2005cvm}: a subset of training points that are sufficient for approximately optimizing the objective corresponding to the full training set. They approximately solve a Minimum Enclosing Ball (MEB) problem that they show is equivalent to the SVM. In the language of our paper, it shows an algorithm for achieving a $\pm \eps$ additive approximation for the SVM objective in the batch setting using space $O(\eps^{-2})$. \cite{rai2009svm} adapts these ideas to the streaming setting by showing a simple one-pass approximate MEB algorithm. However, they only achieve a \textit{constant} approximation, rather than one with a target error $\eps$. To the best of our knowledge, we are the first to analyze streaming SVM algorithms with sub-constant approximation guarantees.

\subsection{Techniques}
\subparagraph*{Multiplicative approximation algorithm.} 
First note that by considering the points labeled $+1$ and $-1$ separately, the point estimation problem that we consider reduces to the following: given a set of $n$ points, the goal is to sketch them such that later, given a query hyperplane denoted by $H = (\theta , b)$, one can estimate the sum of distances of the points on one side of $H$ to $H$. In one dimension, the points are positioned on the real line and the query is also a value $b$ on the real line and the goal is to compute $\sum_i \max\{0 , b-x_i\}$.

The idea behind the streaming algorithm is the following: we uniformly sample a $1/2^i$ fraction of the input points, for $i = 1, 2, \ldots, O(\log n)$, and for each of the $O(\log n)$ sampling rates, we store the smallest $\tilde{O}(\epsilon^{-2})$ points that we have seen. Now given a query point $q$, we want to estimate the sum of distances of points $p < q$ to $q$. If there are fewer than $\tilde{O}(\epsilon^{-2})$ input points $p$ less than $q$, then we have stored all of them and can compute this sum exactly. Otherwise we can think of partitioning the input points $p$ into geometric scales, in powers of $2$, based on their distance to $q$. The main insight is if one of these scales contributes to the overall sum, it must have a number $n'$ of points in it of roughly the same order as the total number of points $p'$ even further away from $q$. This is because each such point $p'$ contributes even more than any point $p$ in this scale to the sum. Consequently, if we choose $i$ for which $1/2^i \approx 1/(n' \epsilon^2)$, then there will be about $\tilde{\Theta}(1/\epsilon^2)$ survivors in the sampling that are at this distance scale, and we will have stored all of them. 
By separately estimating the contribution of each scale and adding them together, we obtain our overall estimate. We note that to achieve our overall $\tilde{O}(\epsilon^{-2})$ space bound we need to obtain as crude of an additive error as possible for scales that do not contribute as much as other scales to the overall objective. We do this by separately first estimating the contribution of each scale up to a constant factor, and then refining it appropriately. 


\subparagraph*{Additive approximation algorithms.} We also show {\em additive} approximation algorithms in the low dimensional regimes. As our lower bound depends on the diameter of the points, we assume without loss of generality that our point set has diameter $1$. For $d=1$,  the sketching algorithm groups the points into $2/\sqrt \eps$ groups in a way that i) each group has diameter at most $\sqrt \eps$, and ii) each group has at most $n\sqrt \eps$ points in it. It is easy to verify that such a grouping results in an additive approximation of $\eps$. To make the algorithm work in the streaming setting, we must maintain the groups as the points arrive. Note that the partitioning based on the diameter can be done in advance. However in order to also partition based on the number of points (and make sure each group gets at most $n\sqrt \eps$ points in it), we create a binary tree for each original group (that has diameter $\sqrt \eps$). Whenever a group reaches its maximum size of $n\sqrt \eps$, we create two children corresponding to that group and further points arriving in the group will be assigned to one of the children. Note that it is important that we cannot partition the previous points into the two children as we already discarded them in the stream.

For $d=2$, we maintain a quad tree on the points in $[-1,1]\times[-1,1]$. Whenever a cell gets too many points we further partition it into four until its side length becomes too small. We process every cell of the quad tree so that if the query (which is a line now) does not collide with the cell, we can exactly compute the sum of distances of the points in the cell to the line (using \cite{pass-glm}). Ignoring the points in the cells that the line crosses, will result in an algorithm with space usage of $O(\eps^{-1})$. To push it down to $O(\eps^{-4/5})$, we randomly sample a point from each cell and for the crossing cells, we will use the single sampled point to estimate the average distances of the points in the cell to the query line.

Finally, for the lower bounds in low dimensions, we develop reductions from the {\sc Indexing} problem. We show how to consider $k = \Omega(\eps^{-\frac{d+1}{d+3}})$ positions inside a $d$-dimensional sphere, where they correspond to a bit-string of length $k$ by Alice. If her $i$-th bit in the string is $1$, then she will put $n/k$ actual points in the corresponding positions, otherwise she does not put any point there. We show that Bob can recover Alice's input using hyper-plane queries.

\subparagraph*{Lower bound in high dimensions.} We let $b=0$ so that $F(\theta,0) = \frac{1}{n}\sum_i \max\{0,\theta\cdot x_i\}$. This is similar to the subspace sketch problem studied in~\cite{LWW20}, which considers approximating $\sum_i \phi(\theta\cdot x_i)$ for $\phi(t) = |t|$ up to a multiplicative factor of $(1+\eps)$. Here we have $\phi(t) = \max\{0,t\}$ instead (by flipping $\theta$) and an additive error $\eps n$. The proof in~\cite{LWW20} turns the multiplicative error into an additive error and so we can adapt the same approach in our current case. Following the same approach, we can show an $\Omega(d/(\eps^2\polylog(1/\eps)))$ lower bound when $d=\Omega(\log(1/\eps))$ for the point estimation problem. Below we sketch the main idea for the proof, which is similar to that in \cite{LWW20}. 

We show an $\Tilde{\Omega}(1/\eps^2)$ lower bound for $d=\Theta(\log(1/\eps))$. The lower bound for general $d$ follows from the concatenation of $\Theta(d/\log(1/\eps))$ independent smaller hard instances.

In the remainder of this subsection let $d=\Theta(\log(1/\eps))$ be such that $n=2^d = \Tilde{\Theta}(1/\eps)$. Consider all the $\{-1,1\}^d$ vectors and let $x_i$ be the $i$-th $\{-1,1\}$-vector scaled by some scalar $r_i$. Define a matrix $M\in \R^{n\times n}$, indexed by $\{-1,1\}$-vectors, as $M_{ij} = \phi(\langle i,j\rangle)$. Then we have for $\theta\in\{-1,1\}^d$ that $\sum_i \phi(-\langle \theta, x_i\rangle) = \sum_i \phi(\langle \theta, x_i\rangle) = \sum_i \phi(r_i \langle \theta, i\rangle) = \sum_i r_i \phi(\langle \theta, i\rangle) = \sum_i M_{\theta, i}r_i = \langle M_\theta,r\rangle$ if all $r_i\geq 0$, where $M_\theta$ denotes the $\theta$-th row of $M$. Our goal is to encode random bits $s_i$ in the scalars $r_i$, such that obtaining $\langle M, \theta\rangle$ within additive $\eps n$ allows us to recover as many bits $s_i$ as possible. 

First we allow $r_i$ to be negative and consider $(Mr)_\theta$. Let
$r = \sum_i s_i \cdot \frac{M_i}{\|M_i\|_2},
$
where $s_1,\dots,s_n$ are i.i.d.\ Rademacher random variables. It follows from a standard concentration inequality that $\|r\|_\infty \leq \poly(d)$. If $M$ had orthogonal rows, then $\langle M_\theta, r\rangle = s_\theta \|M_\theta\|_2$, in which case we can recover the bit $s_\theta$ from the sign of $\langle M_\theta,r\rangle$, provided that $\|M_\theta\|_2$ is larger than the additive error $\eps n$. 

However, $M$ does not have orthogonal rows. The argument above still goes through so long as we can identify a subset $\mathcal{R} \subseteq [n] = [2^d]$ of size $|\mathcal{R}| = \Omega(2^d / \poly(d))$ such that the rows $\{M_i\}_{i\in\mathcal{R}}$ are nearly orthogonal, meaning that the $\ell_2$ norm of the orthogonal projection of $M_i$ onto the subspace spanned by other rows $\{M_j\}_{j \in \mathcal{R} \setminus \{i\}}$ is much smaller than $\|M_i\|_2$. To this end, we study the spectrum of $M$ using Fourier analysis on the hypercube, which shows that the eigenvectors of $M$ are the rows of the normalized Hadamard matrix, while the eigenvalues of $M$ are the Fourier coefficients associated with the function $g(s) = \phi(d - 2w_H(s))$, where $w_H(s)$ is the Hamming weight of a vector $s \in \{0,1\}^d$. It can be shown that $M$ has at least $\Omega(2^d / \poly(d))$ eigenvalues of magnitude $\Omega(2^{d  / 2} / \poly(d))$. For the $\theta$'s which correspond to those eigenvalues, we have $\|M_\theta\|_2 = \Omega(2^{d/2}/\poly(d))$ so that $\|M_\theta\|_2 = \Omega(\eps n)$ for our choice of $n$ and $d$, as required by the argument.

Recall that we require $r_i\geq 0$. Since $\|r\|_{\infty} \le \poly(d)$ with high probability, we can just shift each coordinate of $r$ by a fixed amount of $\poly(d)$ to ensure that all entries of $r$ are positive. We can still obtain $\langle M_\theta, r \rangle$ with an additive error $O(\eps 2^d\poly(d))$, since the amount of the shift is fixed and bounded by $\poly(d)$.

Last, rescaling $\theta$ and $x_i$'s to unit vectors loses $\poly(d) = \poly(\log(1/\eps))$ factors in the lower bound and we continue to have an $\Omega(1/(\eps^2\polylog(1/\eps)))$ lower bound.

\subparagraph*{Optimization lower bound.}
We prove optimization lower bounds by reducing from the {\sc Indexing} problem: Alice encodes points at specific locations on the unit sphere, then Bob uses the optimization sketch to decode whether a point was added at some particular location. The challenge is that Bob must be able to reason about a \emph{single} data point when given access to an approximate optimum corresponding to an \textit{entire} dataset. The key idea in getting this to work is for Bob to add some additional points with the following property: if the location being queried does not contain a point, then the added points are the \textit{only} support vectors, i.e. $(\theta^*, b^*)$ is entirely determined by the added points, whereas if the the location being queried does contain a point, then $(\theta^*, b^*)$ is entirely determined by the added points and the point at that location. Hence, $(\theta^*, b^*)$ can take on exactly two possible values, and does not depend on any of the remaining points in the dataset. Moreover, the strong convexity of $F_\lambda$ allows us to upper bound $\|(\theta^*, b^*) - (\hat{\theta}, \hat{b})\|$ in terms of $\eps$. We make this precise in the following lemma, which we will refer to multiple times later on:

\begin{lemma}\label{lem:str_conv}
If $F_\lambda(\hat{\theta}, \hat{b}) \leq F_\lambda(\theta^*, b^*) + \eps$, then $\|(\hat{\theta},\hat{b}) - (\theta^*,b^*)\|_2 \leq \sqrt{2\eps/\lambda}$.
\end{lemma}

By exactly characterizing what $(\theta^*, b^*)$ is in these two possible cases, and showing that the gap is more than $2\sqrt{2\eps/\lambda}$, we show that Bob can distinguish these two situations and decode the bit. The analysis is delicate, requiring a carefully chosen construction for the proof to go through.




\subparagraph*{Open Problems.}
For the case of $d=1$, we have matching upper and lower bounds for (additive) point estimation of $\Theta(\eps^{-1/2})$. This also translates into an \textit{optimization} upper bound of $O(\eps^{-1/2})$, the best upper bound we know of for this setting. However, we have no optimization lower bound for $d=1$. Instead, we have an optimization lower bound of $\Omega(\eps^{-1/4})$ for $d=2$ and $\Omega(\eps^{-1/2})$ for $d \ge 3$. Moreover, for optimization in high dimensions, there remains a gap between the $\Omega(\eps^{-1/2})$ lower bound and $O(1/\eps)$ upper bound. It remains to close all of these gaps.

Also, while in this work we focus on linear SVMs, often times non-linear kernels are preferred in practice. This raises the question of whether we can extend our results to this setting. One approach is to use random feature maps that allow one to convert a kernel SVM problem into a linear SVM problem \cite{rahimi2007features}. However, this increases the dimension significantly, so that sampling and running SGD is more efficient than optimization via point estimation.\footnote{As described in \cite{shalev-shwartz2007pegasos}, one can use adapt SGD to work for kernelized SVMs. This involves tracking dual variables $\alpha_i$, which we can do with the same space complexity of $O(\frac{d}{\lambda\eps})$ as before.}

\section{Point Estimation}
In this section, we study the streaming complexity of the point
estimation problem. Specifically, the algorithm
sees the data points $(x_i,y_i)$, for $1\leq i\leq n$, one by one.
The goal is to keep a sketch of the data such that later, given the
query parameters $(\theta,b)$, it can output an estimate of the
SVM objective function $F(\theta,b)$.

\subparagraph*{Setup for point estimation.} We can simplify the presentation
by focusing on a slight simplification of the SVM objective (without
loss of generality). First, we note that, since $y_i\in\{+1,-1\}$, we
can estimate the contribution from $x_i$'s with $y_i=+1$ and $y_i=-1$
separately. Hence, for point estimation it is enough to assume that
$y_i=+1$, as well as that $\lambda=0$ (since the regularization can be
computed independently of the data). Furthermore, we can just work
with the following related objective:
\begin{equation}
\label{eqn:ridgeF}
F(\theta,b) := \frac{1}{n}\sum_{i=1}^n \max\{0, b-\theta \cdot x_i\},
\end{equation}
by adjusting $b$ accordingly. Hence we focus on estimating the
function from Eqn.~\eqref{eqn:ridgeF} for the rest of this section.



Note that when $d=1$, it is enough to consider the case of $\theta=+1$. First,
because for $\theta\in(0,1)$, we can rescale the output of a sketch
(that uses $\theta=+1$ and accordingly rescaled $b$). Second, because
the case of $\theta\in[-1,0)$ is precisely symmetric, so one can just keep two
sketches, one for each of $\theta\in\{-1,+1\}$. Note that the
objective simplifies to $F=\tfrac{1}{n}\sum_{i=1}^n \max\{0,q-x_i\}$
where $q=b$. We will call the value $q$ the query.

\subsection{Multiplicative $(1+\eps)$ approximation algorithm for $d=1$}

Here we consider the case of $d=1$: We are given a set of $n$ points $x_i \in \Re$, and given any query $q\in \Re$, the goal is to approximate $\sum_{i: x_i\leq q}
(q-x_i)$ up to a multiplicative $1+\eps$ factor. To analyze the bit complexity of the problem, we assume the points are integers between $1$ and $W$. A simple sketching algorithm is given in \cref{apx:pe_ub_sketch}. 
Here we present a streaming algorithm for the problem.

\def \Val{V}

\subparagraph*{Streaming.}
Here we assume that the values $x_1,\dots,x_n$ are given in a stream in this order, and we are allowed to make a single pass over it, and the query $q$ is given at the end of the stream. Note that $x_i$'s are not necessarily sorted, and for simplicity, we assume all $x_i$'s are distinct. The algorithm maintains the following sketch throughout the stream.

\subparagraph*{Sketch.} The sketch consists of two collections of sets of samples as described below. The first collection is used to get a crude (constant factor) approximation of the contribution of each contributing interval as defined later, and the second collection is used for a more precise approximation.
\begin{itemize}
\item For each $0\leq i \leq \log n$, sample every point with probability $1/2^i$, and preserve the $m_1 = C_1\eps^{-1} \log^2 W$ sampled points with the smallest $x$ value in the set $E_i$, where $C_1$ is a constant to be specified later.
\item For each $0\leq i \leq \log n$, sample every point with probability $1/2^i$, and preserve the $m_2 = C_2\eps^{-2} \log W$ sampled points with the smallest $x$ value in the set $S_i$, where $C_2$ is a constant to be specified later.
\end{itemize}

\begin{observation}[Space]\label{obs-1dmult-space} The sets $E_i$ and $S_i$ can be maintained in a stream. Let $M=\max\{m_1,m_2\}$, then the space usage of the algorithm is $O(\log n \cdot M \cdot \log W) = O(\frac{\log n \cdot \log^2 W}{\eps}(\frac{1}{\eps} + \log W))$ bits.
\end{observation}

Next we describe and analyze the query processing algorithm. 

\subparagraph*{Query algorithm.} 
Let $p$ be the largest value in $S_0\cup E_0$. Given the query point $q\in \mathbb{R}$, we proceed as follows.
\begin{itemize}
\item  If $q\leq p$, then we can report an exact solution using the corresponding sample set: e.g. if $p\in S_0$, then we output  $\sum_{x\in S_0, x\leq q} (q-x)$. 
\item Otherwise, we group the points based on their distance to $q$ and estimate the contribution of each group separately. More precisely, let $D = q-p$, which is a positive number, and for each $1\leq j\leq \log D$, define the interval  $R_j = (q - \frac{D}{2^{j-1}} , q - \frac{D}{2^j}]$. For notational convenience, let $R_0$ be the interval covering $S_0\cup E_0$ which ends at the point $p$.
Finally for each $0\leq j\leq \log D$, let $t_j=|P\cap R_j|$ be the number of points in the interval $R_j$, and $T_j=|P\cap (\bigcup_{k<j}R_k)|$ be the number of points to the left of the interval $R_j$.
\item Let $i'(j)$ be the largest $i$ such that $E_i$ contains at least $\log D$ points from $R_j$. If no such $i'$ exists, let $i'(j)=-1$. The value $i'(j)$ shows which sampled set ($E_{i'(j)}$) we should use for our crude approximation. As we show in Lemma \ref{clm-mult-1}, if $i'(j)=-1$, the contribution of the points in $R_j$ can be ignored).

\item Let $\phi_j = \min\{ 1, \frac{|E_{i'(j)}\cap R_j|}{|E_{i'(j)}\cap (\cup_{k< j}R_k)|}\}$.
This value is used to approximate the ratio of the points in $R_j$ to the points that are to the left of $R_j$, i.e., $\phi_j\approx\frac{t_j}{T_j}$. This is verified in Lemma \ref{clm-mult-2}.

\item We set the value of $i(j)$ as follows.
\begin{itemize}
\item If $i'(j)=-1$ or $\phi_j\leq \frac{\eps}{\log W}$, then set $i(j)=-1$. In this case, the contribution of $R_j$ can be ignored.
\item Otherwise, if $\phi_j\geq \frac{1}{\log W}$, then set $i(j)$ to be the largest $i$ such that $S_i$ contains at least $1/\eps^2$ points from $R_j$. If no such $i$ exists let $i(j)=-1$. This case in analyzed in Lemma \ref{clm-mult-3} and Lemma \ref{clm-mult-4}.
\item Finally, if $\frac{\eps}{\log W}\leq \phi_j\leq \frac{1}{\log W}$,  then set $i(j)$ to be the largest $i$ such that $S_i$ contains at least $(\phi_j\log W/\eps)^2$ points from $R_j$. This case in analyzed in Lemma \ref{clm-mult-5} and Lemma \ref{clm-mult-6}.
\end{itemize}
\item Report $\sum_{j\leq \log D, i(j)\neq -1} \sum_{x\in S_{i(j)}\cap R_j} 2^{i(j)} (q-x)$. That is for all sets whose contribution is significant (equivalently $i(j)\neq -1$) we estimate their contribution using sample set $S_{i(j)}$.
\end{itemize}

We then have the following.

\begin{lemma}[main lemma]\label{lm-main}
This algorithm returns a $(1+O(\eps))$ multiplicative approximation.
\end{lemma}

First, it is clear from the algorithm description that in the case of $q\leq p$, the algorithm produces an exact solution. To show the correctness in the otherwise case, as stated in the description let $t_j$ be the number of points in $P$ that fall in the interval $R_j$, and let $T_j$ be the number of points that fall to the left of $R_j$, i.e., $T_j = |\{x\in P \colon x\leq q-(D/2^{j-1})\}|=|\bigcup_{k<j} R_k|$ and note that $T_j\geq \max\{|S_0|,|E_0|\} \geq \frac{2\log W}{\eps}\max\{\log W + 1/\eps\}$. Moreover, let $\Val = \sum_{x\in P, x\leq q} (q-x)$ be the value of the solution and $\Val_j =\sum_{x\in P\cap R_j} (q-x)$ be the contribution of the points in $R_j$ to the solution. 

\begin{lemma}[If Not Enough Samples, Set Doesn't Contribute]\label{clm-mult-1} 
If $t_j< \log D$, then $i'(j)=-1$. Moreover, if $i'(j)=-1$, then with high probability $\frac{t_j}{T_j}\leq \frac{2\eps}{\log W}$.
\end{lemma}
\begin{proof}
For the first claim, note that for such $j$, even if we sample the points with probability $1$, we don't get $\log D$ points in $E_i$, and thus $i'(j)=-1$.

Now suppose $i'(j)=-1$. If $t_j\leq 2\log D$, then since $T_j\geq \log^2 W/\eps$, we have that $\frac{t_j}{T_j}\leq \frac{2\eps}{\log W}$.
Otherwise, let $i'\geq 0$ be such that $2^{-(i'+1)}\leq \frac{2\log D}{t_j}\leq 2^{-i'}$. Now, if we sample every point in $R_j$ with probability $2^{-i'}$, in expectation, 
we sample $t_j2^{-i'}$ points which is between 
$2\log D$ and $4\log D$. Moreover, since $t_j\geq 2\log D$, we can use Chernoff bound, proving that with high probability the number of samples is between $\log D$ and $8\log D$. 
Therefore, since the final sample $E_{i'}$ does not contain $\log D$ points from $R_j$, it means that it contains at least $m_1 - 8\log D \geq m_1/2 = \frac{C_1\log^2 W}{2\eps}$ points from $\cup_{k<j}R_k$ (for sufficiently large constant $C_1$).

On the other hand, the expected number of sampled points from $\bigcup_{k<j}R_k$ is $T_j2^{-i'}$ which with high probability (using Chernoff again) should be at least $m_1/4 \geq \frac{C_1\log^2 W}{4\eps}$. Therefore, we get that 
$$T_j\geq 2^{i'}\frac{C_1\log^2 W}{4\eps} \geq \frac{t_j}{4\log D}\cdot\frac{C_1\log^2 W}{4\eps}\geq \frac{C_1 t_j \log W}{16 \eps}$$

Thus for $C_1\geq 8$, we get that $\frac{t_j}{T_j}\leq\frac{2\eps}{\log W}$
\end{proof}

\begin{lemma} [If Enough Samples, Get Initial Constant Factor Approximation] \label{clm-mult-2}
If $i'(j)\neq -1$, then $\phi_j$ approximates $\min\{1,\frac{t_j}{T_j}\}$ by a constant factor.
\end{lemma}
\begin{proof}
By Lemma \ref{clm-mult-1}, we know that $t_j \geq \log D$. 
\begin{itemize}
\item First note that if $T_j\leq t_j$, with high probability the number of sampled points from $\cup_{k<j} R_k$ is less than $2$ times the number of sampled points from $R_j$, and therefore, $\phi_j\geq 1/2$, and the lemma is proved.

\item Second, note that if $\frac{t_j}{T_j}\leq \frac{\eps}{2C_1\log W}\cdot \frac{\log D}{\log W}$, then the number of sampled points from $\cup_{k<j} R_k$ with high probability (using Chernoff) is at least $\frac{1}{2}\cdot\frac{T_j}{t_j}\cdot \log D \geq \frac{\log D}{2} \cdot \frac{2C_1 \log^2 W}{\log D\eps} \geq \frac{C_1 \log^2 W}{\eps} = m_1$ which is a contradiction, because then we would not have picked $\log D$ points in $E_{i'(j)}$ from $R_j$.

\item 
In the otherwise case, we show that the sample $E_{i'(j)}$ suffices to get a constant factor approximation to both values $t_j$ and $T_j$ (and hence, $t_j/T_j$). First, to see the latter, note that $T_j\geq t_j$ and therefore with high probability (using Chernoff) we get at least $\frac{\log D}{2}$ samples from $\cup_{k<j}R_k$ in the set $E_{i'(j)}$ which are chosen uniformly at random. This is enough for computing a constant factor approximation of $T_j$ with high probability.

Second, if $E_{i'(j)}$ only contains samples from the first $t_j/8$ fraction of the points in $R_j$, then with high probability, there would still be $\log D$ samples from $R_j$ in $E_{i'(j)+1}$ (this is because with high probability, we only get less points from $\bigcup_{k<j} R_k$ in $E_{i'(j)+1}$ and still at least $\log D$ points from $R_j$ in it). However, this contradicts the choice of $i'(j)$. Therefore, we get a uniform sample of size $\Omega(\log D)$, from a constant fraction of the points in $R_j$, meaning that we can approximate $t_j$ up to a constant factor with high probability. Therefore, $\phi_j$ will be a constant approximation to the value of $\min \{1,\frac{t_j}{T_j}\}$. 
\end{itemize}
\end{proof}


The following two lemmas analyze the case of $\phi_j\geq \frac{1}{\log W}$.
\begin{lemma}[Enough Samples are Found From Large Contributing Sets]\label{clm-mult-3}
If $i'(j)\neq -1$ and $\phi_j\geq \frac{1}{\log W}$, then 
we have $i(j)\neq -1$.
\end{lemma}
\begin{proof}
Note that by Lemma \ref{clm-mult-2}, $\phi_j$ is a constant factor approximation for $\min\{1,t_j/T_j\}$ which means that either $t_j\geq \Omega(T_j)$ or $t_j\geq \Omega(T_j/\log W)$ as $\phi_j\geq 1/\log W$. Let $C_3$ be the constant in this inequality, i.e., $t_j\geq  T_j/(C_3\log W)$. Moreover, as $T_j\geq |S_0|= \frac{C_2\log W}{\eps^2}$, we get that $t_j\geq \frac{C_2}{C_3\eps^2}\geq 2/\eps^2$ for $C_2\geq 2C_3$.


Now let $i$ be such that $2^{-(i+1)} \leq 2/(\eps^2 t_j) \leq 2^{-i}$. This means that by sampling the points with rate $2^{-i}$, in expectation, we sample $2^{-i}t_j \geq 2/\eps^2$ points from $R_j$ and moreover, with high probability we will sample at least $1/\eps^2$ points from $R_j$ (here we used the fact that since $i'(j)\neq -1$, we have $t_j\geq \log D$ and thus we can apply Chernoff bound).
Furthermore, in expectation, we will sample $2^{-i} T_j \leq 4 T_j / (t_j \eps^2) \leq 4C_3\cdot \log W /\eps^2$ points from intervals $R_0,\dots, R_{j-1}$. Therefore, since we keep $m_2 = C_2 \log W / \eps^2$ smallest sampled points in $S_i$, by choosing $C_2$ large enough, i.e., $C_2\geq (4C_3+2)$, we will store all sampled points of $R_j$ in $S_i$ as well. Therefore, $i(j)\neq -1$.
\end{proof}

\begin{lemma}[Large Contributing Sets have Small Relative Error]\label{clm-mult-4}
If $i'(j)\neq -1$ and $\phi_j\geq \frac{1}{\log W}$, then 
we have that $\sum_{x\in S_{i(j)}\cap R_j} 2^{i(j)} (q-x)$ is a $(1+2\eps)$ approximation of $\Val_j$.
\end{lemma}
\begin{proof}
As in the previous lemma, let $i$ be such that $2^{-(i+1)} \leq  2/(\eps^2 t_j)	\leq 2^{-i}$. Therefore, by similar arguments to the above lemma, we know that $i(j)\geq i$, which further means that all sampled points from $R_j$ are kept in $S_{i(j)}$. We thus get a uniform sample of size at least $1/\eps^2$ from the interval $R_j$, which is enough for an additive $\eps t_j (D/2^{j-1})$ approximation as every point in the interval is contributing between $D/2^{j}$ and $D/2^{j-1}$. On the other hand, we have that $\Val_j \geq t_j D/2^{j}$. Thus this additive approximation translates to a $(1+2\eps)$ multiplicative approximation.
\end{proof}

The following two lemmas analyze the case of $\frac{\eps}{\log W}\leq \phi_j\leq \frac{1}{\log W}$.
\begin{lemma}[Enough Samples are Found from Small  Contributing Sets] \label{clm-mult-5}
If $i'(j)\neq -1$ and $\frac{\eps}{\log W}\leq \phi_j\leq \frac{1}{\log W}$, then if $\Val_j \geq \Val (\eps/\log W)$, we have $i(j)\neq -1$.
\end{lemma}
\begin{proof}
Note that $V_j\geq V(\eps/\log W)$ means that $t_j \geq T_j (\eps/\log W)$, as otherwise if $t_j <T_j (\eps/\log W)$, we have that the total contribution of $R_j$ is at most $\Val_j \leq t_j D / 2^{j-1}$, while we have that $\Val \geq T_j D/2^{j-1}$; which means that $\Val_j \leq T_j(\eps / \log W)  D / 2^{j-1} \leq (\eps /\log W)\Val$, which is a contradiction.

Now let $i$ be such that $2^{-(i+1)} \leq 2(\phi_j \log W)^2/(\eps^2 t_j) \leq 2^{-i}$ \footnote{Note that again by the conditions of the lemma on $\phi_j$, the fact that $\phi_j$ approximates $t_j/T_j$, and the fact that $T_j\geq C_2\log W /\eps^2$, this value is always at most $1$ and therefore such $i$ exists.}. This means that in expectation, we sample $2^{-i}t_j \geq 2(\phi_j \log W/\eps)^2$ points from $R_j$ and moreover, with high probability we will sample at least $(\phi_j \log W/\eps)^2$ points from $R_j$.
Furthermore, in expectation, we will sample $2^{-i} T_j \leq 4 T_j (\phi_j \log W)^2 / (t_j \eps^2) \leq 4\phi_j \log^2 W /\eps^2 \leq 2C_3\log W/\eps^2$  points from intervals $R_0,\dots, R_{j-1}$, where in the last inequality we used the fact that $\phi_j\leq 1/\log W$, and $C_3$ is a constant with which $\phi_j$ approximates $t_j/T_j$. Therefore, since we keep $m_2 \geq (C_3+2)\log W / \eps^2$ smallest points in $S_i$, we will store all sampled points of $R_j$ in $S_i$. Therefore, $i(j)\neq -1$.
\end{proof}

\begin{lemma}[Small Contributing Sets have Small Additive Error]\label{clm-mult-6}
If $i'(j)\neq -1$ and $\frac{\eps}{\log W}\leq \phi_j\leq \frac{1}{\log W}$, then if $\Val_j \geq \Val (\eps/\log W)$, we have that $\sum_{x\in S_{i(j)}\cap R_j} 2^i (q-x)$ is a $\frac{\eps}{\log W}$ additive approximation of $\Val$.
\end{lemma}
\begin{proof}
As in the previous lemma, let $i$ be such that $2^{-(i+1)} \leq  2(\phi_j \log W)^2/(\eps^2 t_j)	\leq 2^{-i}$. Therefore, by similar arguments to the above lemma, we know that $i(j)\geq i$, which further means that all sampled points from $R_j$ are kept in $S_{i(j)}$. We thus get a uniform sample of size at least $(\phi_j \log W/\eps)^2$ from the interval $R_j$, which is enough for an additive $\frac{\eps}{\phi_j\log W} t_j (D/2^{j-1})$ approximation as every point in the interval is contributing between $D/2^{j}$ and $D/2^{j-1}$. However, as $\Val \geq T_j D/2^{j-1}$, this translates to a $\frac{\eps}{\phi_j \log W}\frac{t_j}{T_j}\Val= \frac{\eps}{\log W}\Val$ additive approximation. 
\end{proof}

We now prove \Cref{lm-main}, which states that the algorithm returns a $(1+O(\eps))$ multiplicative approximation.
\begin{proof}
Let us consider the following cases separately.
\begin{itemize}
\item For $j=0$, we get the exact contribution of the points in $R_0$.
\item For $j$ where $i'(j)=-1$, we know by lemma \ref{clm-mult-1} that $\frac{t_j}{T_j}\leq 2\eps/\log W$. Therefore the total contribution of all $\Val_j$ for such $j$ is at most $2\eps\Val$.
\item Moreover, whenever $i(j)=-1$, using Lemmas \ref{clm-mult-3} and \ref{clm-mult-5}, we know that $\Val_j \leq \Val (\eps/\log W)$. Summing over all such $j$, we get a total additive error of $\eps\Val$.
\item Now consider all $j\geq 1$ such that $\phi_j \geq \frac{1}{\log W}$. By Lemma \ref{clm-mult-3} and Lemma \ref{clm-mult-4}, we get a $(1+2\eps)$ multiplicative approximation of their contribution. 

\item Finally consider all $j\geq 1$ such that $\frac{\eps}{\log W} \leq \phi_j \leq \frac{1}{\log W}$. If  $\Val_j \geq \Val (\eps/\log W)$, by lemma \ref{clm-mult-6}, we get an additive  $\frac{\eps}{\log W}\Val$ approximation of their contribution. Summing over all such $j$, this will give a $(1+\eps)$ multiplicative approximation.

However if $\Val_j < \Val (\eps/\log W)$ and  $i(j)\neq -1$ for such $j$, it means that we have sampled points in $R_j$ with probability $2^{-i(j)}$ and potentially kept only some of them in $S_{i(j)}$ this only causes an under estimation of the contribution of $R_j$. Note that, the samples that the algorithm has chosen to keep in $S_{i(j)}$ might be biased towards the smaller end of the interval, i.e., $D/2^{j-1}$, however this can only cause an over estimation by a factor of $2$. Therefore, for such a $j$, we have $0\leq \sum_{x\in S_{i(j)}\cap R_j} 2^{i(j)} (q-x) \leq 2(1+\eps)\Val_j$. Again because $\Val_j\leq \Val \eps /\log W$, the total error of such $j$ will be at most a multiplicative $(1+\eps)$ factor.
\end{itemize}
\end{proof}

We then have the following corollary.

\begin{corollary}
There exists a one pass streaming algorithm that computes a $(1+\eps)$ multiplicative approximation for point estimation variant of the problem in one dimensional case. Moreover, if the points come from $[W]$, the space usage of the algorithm is $O(\frac{\log^2 n\cdot \log W}{\eps}(\log n+1/\eps))$ bits.
\end{corollary}
\begin{proof}
Note that in the above algorithm we do not need to consider $R_j$ for which $j\geq \log n^2=2\log n$, as the overall contribution of such points to the solution is as most $D/n$ whereas the value of the solution is at least $D$. Thus we can bound one of the $\log W$ in the bound of Observation \ref{obs-1dmult-space} by $\log n$.
\end{proof}

\subsection{Lower Bounds}
We now show that one cannot hope to get a sketch with
multiplicative approximation in higher dimensions than one with a
bound independent of $n$. In fact we show the following additive
approximation lower bound for any sketching algorithm (and hence
streaming algorithm as well).

\begin{theorem}
For any $d\geq 1$, $\eps\in(0,1)$, and $n\ge 1/\eps$, there exists an
instance of the problem, where the point set has diameter $O(1)$, such that getting an algorithm with additive
approximation factor $\eps$, requires space of
$\Omega(\eps^{-(d+1)/(d+3)})$ bits.

Moreover, getting a $(1+\eps)$-multiplicative approximation for $d\geq
2$ requires $\Omega(n)$ space.
\end{theorem}

We note that while this theorem formally applies to the simplified
objective, it immediately translates to the SVM objective as well as
we only use points with one label (when the problems are exactly
equivalent).

\begin{proof}
We prove this theorem by a reduction from the standard {\sc Augmented
  Indexing} problem.  In this problem, Alice is given a bit string $s$ of
length $m$ and Bob has an index $i\in [m]$ as well as bits $s_1,\ldots,
s_{i-1}$. The goal is for Alice to send a message to Bob so that he can
recover $s_i$, the $i$th bit in Alice's string. It is a standard fact that
this requires $\Omega(m)$ bits.

\subparagraph*{Case of $d=1$.} First consider the one dimensional case and
let $r = 1/\sqrt \eps$. Suppose that Alice holds a string of length
$r$, termed $s_1\ldots s_r$. From that, she will construct an instance of our point estimation
sketching problem. For each $0\leq i < r$, if her $i$th bit is one,
she will put $n/r = n\sqrt \eps$ points in position $3i/r =  3i\sqrt
\eps$. Otherwise if the bit is $0$, she will not put any point
there. Thus all points will be positioned in $[0,1)$ with the diameter
  less than $1$.

  To learn the $i$th bit, Bob will query the presumed sketch with $b = 3(i+1)\sqrt
  \eps$, obtaining a value $v$. Bob will
  subtract the contribution from the points associated with the first $i-1$ bits.
  Note that the resulting value is $\tfrac{1}{n}\cdot s_i\cdot
  n\sqrt{\eps}\cdot 3\sqrt{\eps}$, up to an additive error
  $\eps$. Therefore it is possible to recover the value of the encoded bit. Hence
  the one-dimensional point estimation problem cannot be solved in
  space less than $\Omega(r) = \Omega(1/\sqrt \eps)$.

\subparagraph*{Case of $d=2$.} Next, consider the two dimensional case, and for parameters $s$ and $r$ (to be specified later), consider the $s \times r$ potential positions inside a circle of unit radius. More specifically, for $1\leq i\leq s, 1\leq j\leq r$, the $(i,j)$-th position is the point at angle $2\pi i/s$ and at radius $1- \frac{j-1}{2r}$. These positions correspond to the $sr$ bits in the index problem held by Alice. For the $((j-1) s+i)$-th bit in her bit-string, Alice will put $n/(sr)$ actual points at the $(i,j)$-th position described above iff the corresponding bit is equal to $1$. She will then send her point set to Bob. 

Bob can recover any bit of Alice using hyperplane
queries. Specifically, in order to figure out the $((j-1)s+i)$-th bit
of Alice, Bob can ask the hyper-plane corresponding to $\theta = 2\pi
i / s$ and $b= 1-j/(2r)$; and subtract the contribution of points
corresponding to bits up to $((j-1)s+i)$. Note that, if there is no
point at that location, the result should be $0$, otherwise it should be $(n/(sr))\cdot (1/(2r))$. Thus, if the algorithm has an additive approximation less than $\tfrac{1}{3}(n/(sr))
\cdot (1/(2r))$, it can correctly recover the respective bit in Alice's
input.

Also, note that the above means that no multiplicative approximation is possible unless Alice sends $s$ bits. But setting $s=n$ and $r=1$, this will require Alice to send all her input to Bob. 


Note that for the hyperplane to include only the $i$th point from the
tier $j$ circle, we need to set $r$ so that $1/(2r)\approx
(1-\cos(2\pi/s)) \approx 2\pi^2 / s^2$. Thus we need additive
approximation $\Theta(n/s^5)$. We set $s = \eps^{-1/5}$, and get that
the total size of the sketch is at least $\Omega(sr) =
\Omega(\eps^{-3/5})$.

\subparagraph*{Case of $d>2$.} To generalize the result to higher
dimensions, we put $s$ points uniformly on the $d-1$ dimensional
unit sphere and repeat this for $r$ different radii as before. More
precisely, using we put an $\ell$-net on the surface of the unit
sphere. It is a standard fact that we can have $s =\Theta(1/\ell)^{d-1}$ points on the surface so that their pairwise distance is at least $\ell$. Similar to the two-dimensional case, we have that $(1/2r) \approx (1-\cos (O(\pi/\ell))) \approx O(1/\ell^2)$. Therefore, the additive approximation the algorithm can tolerate is $\frac{n}{sr}\cdot\frac{1}{2r}$ which should be at most $n\eps$ and therefore, we get that $\eps = 1/(sr^2)$. Inserting the values of $s$ and $r$ using the value of $\ell$, we get that $\eps = \ell^{d+3}$. As the space lower bound for the index problem is $\Omega(sr)$, we get that the space requirement for our problem is $sr = 1/(\eps r) = \ell^{2} / \eps = \eps^{2/(d+3)}/\eps = \eps^{-\frac{d+1}{d+3}}$.
\end{proof}

The preceding theorem gives at most an $\Omega(1/\eps)$ lower bound, leaving a quadratic gap from the simple random sampling algorithm of $\tilde O(d/\eps^2)$ bits. In fact, for high dimensions $d = \Omega(\log(1/\eps))$, we can prove a lower bound of $\Tilde{\Omega}(d/(\eps^2\polylog(1/\eps))$ bits, tight up to logarithmic factors. 


We consider a generalized version of the point estimation problem in the $\ell_p$ norm. We consider instead of \eqref{eqn:ridgeF} the following objective for a constant $p > 0$.
\[
F_p(\theta,b) = \frac{1}{n}\sum_{i=1}^n \left(\max\{0, b-\theta\cdot x_i\}\right)^p.
\]

The main result we are going to prove in this section is the following.
\begin{theorem}[high-dimensional]
Let $p \geq 0$ be a constant. There exist constants $C \in (0,1]$ and $\eps_0 > 0$ that depend only on $p$ such that the following holds. Let $d_0 = 2\log_2(C/(\eps\polylog(1/\eps))$. For any $\eps \in (0,\eps_0)$, $d \geq d_0$, $n\geq (d/d_0) 2^{d_0}$, any algorithm that approximates $F_p(\theta, 0)$ with an additive error $\eps$ and with probability at least $2/3$ requires $\Omega(d/(\eps^2 \polylog(1 / \eps)))$ bits of space.
\end{theorem}

\begin{proof}
First observe that an additive $\eps$ estimation to $F_p(\theta, 0)$ is an additive $(\eps n)$-approximation to $\Phi(-\theta)$, which is defined as
\[
\Phi(\theta) = \sum_{i=1}^n \phi((X\cdot \theta)_i),
\]
where $X$ is an $n\times d$ matrix with rows $x_1,\dots,x_n$ and $\phi(t) = t^p\mathbf{1}_{t\geq 0}$. This is similar to the subspace $\phi$-sketch problem defined in~\cite{LWW20} and we shall use the same approach to prove the lower bound with the change of $\phi(t) = |t|^p$ in~\cite{LWW20} to our new $\phi(t)$. Although the subspace $\phi$-sketch problem considers the multiplicative $(1+\eps)$-approximation to $\|X\theta\|_p^p$, the proof in essence uses the additive error, which was a cruel upper bound for $\eps\|X\theta\|_p^p$ in that paper. In the hard instance, the matrix $X$ is a block diagonal matrix, in which each diagonal block is of dimension $2^{d_0}\times d_0$, whose $j$-th row is $\{-1,1\}^{d_0}$, scaled by some $w_j =\Theta(d_0^{1/p})$. It is suffices to show an $\Omega(2^{d_0}/\polylog(d_0))$ lower bound for each diagonal block. To abuse the notation a bit, we shall denote each diagonal block by $X$. Note that $\|X\theta\|_p^p = O(2^{d_0}(d_0)^p)$ and so if the additive error is at most $O(\eps 2^{d_0}(d_0)^p)$, the lower bound will hold. In our case, it always holds that $0\leq \Phi(\theta) = n = 2^{d_0}$ and thus the proof will go through, provided that we show the matrix $M_{i,j}^{(d,p)} = \phi(\langle i,j\rangle)$ ($i,j\in \{-1,1\}^d$) has $N^{(d)} = \Omega(2^d/\poly(d))$ singular values of magnitude $\Omega(2^d/\poly(d))$. In the proof below, we follow the notation in~\cite{LWW20}.

\subparagraph*{The Case $p\in [0,\infty)\setminus 2\Z$.}
The proof is simple in this case. Following the same argument in \cite[Section 3]{LWW20}, we can show that $M_{i,j}^{(d,p)}$ has at least $N^{(d)}\geq \binom{d}{d/2}$ singular values of magnitude $\Lambda_0^{(d,p)} = \Omega_p(2^{d/2}/\sqrt d)$. One can verify that when $8|d$,
\[
\left| \sum_{\substack{0\leq i\leq d\\ \text{$i$ is even}}} (-1)^{i/2} \binom{d / 2}{i / 2}  \phi(d - 2i)\right| = \frac12\left| \sum_{\substack{0\leq i\leq d\\ \text{$i$ is even}}} (-1)^{i/2} \binom{d / 2}{i / 2}  |d - 2i|^p\right|.
\]
The right-hand side is exactly the quantity of interest in \cite[Lemma 3.2]{LWW20} and our claim follows from \cite[Lemma 3.4]{LWW20}. 

\subparagraph*{The Case $p\in 2\Z^+$.}
We modify \cite[Lemma 3.2]{LWW20} and consider instead the Fourier coefficients $\hat g(s)$ for $s\in \mathbb{F}_2^d$ with Hamming weight $d/2-1$, which is the same for all $\binom{d}{d/2-1}$ such $s$'s. Note that
\[
\hat g(s) = \sum_{i=0}^{d} \sum_{j=0}^{i} (-1)^j \binom{\frac{d}{2} - 1}{j}\binom{\frac{d}{2} + 1}{i-j} g(i).
\]
By comparing the coefficients of $x^i$ on both sides of the identity $(1+x)^{d/2-1}(1-x)^{d/2+1} = (1-x^2)^{d/2-1}(1+x)^2$, we see that
\begin{align*} 
\sum_{j=0}^{i} (-1)^j \binom{\frac{d}{2} - 1}{j}\binom{\frac{d}{2} + 1}{i-j} &= 
	\begin{cases}
		\displaystyle (-1)^{\frac{i}{2}}\left[\binom{\frac{d}{2} - 1}{\frac{i}{2}} - \binom{\frac{d}{2}-1}{\frac{i}{2}-1}\right], & i\text{ is even};\\
		\displaystyle 2(-1)^{\frac{i-1}{2}}\binom{\frac{d}{2} - 1}{\frac{i-1}{2}}, & i\text{ is odd}.
	\end{cases} \\
		&= \begin{cases}
		\displaystyle (-1)^{\frac{i}{2}}\frac{d-2i}{i}\binom{\frac{d}{2} - 1}{\frac{i}{2} - 1}, & i\text{ is even};\\
		\displaystyle 2(-1)^{\frac{i-1}{2}}\binom{\frac{d}{2} - 1}{\frac{i-1}{2}}, & i\text{ is odd}.
	\end{cases} 
\end{align*}
Split the summation over $i$ in $\hat g(s)$ into odd and even $i$'s, we have
\begin{align*}
\hat g(s) &= \sum_{\substack{\text{even }i\\ 0\leq i\leq\frac{d}{2}}} (-1)^{\frac{i}{2}}\frac{d-2i}{i}\binom{\frac{d}{2} - 1}{\frac{i}{2} - 1}(d-2i)^p + 2\sum_{\substack{\text{odd }i\\ 0\leq i\leq\frac{d}{2}}} (-1)^{\frac{i-1}{2}}\binom{\frac{d}{2} - 1}{\frac{i-1}{2}}(d-2i)^p\\
&=: S_{\text{even}} + 2S_{\text{odd}}.
\end{align*}
Suppose that $d=4n$ where $n$ is an even integer, we have
\begin{align*}
S_{\text{even}} &= 2^{2p+1} \sum_{\ell=0}^{n-1} (-1)^\ell \frac{n-\ell}{\ell} \binom{2n-1}{\ell-1} (n-\ell)^p \\
&= 2^{2p+1} \sum_{\ell=1}^{n} (-1)^\ell\frac{\ell}{n-\ell} \binom{2n-1}{n-\ell-1} \ell^p \\
&= 2^{2p+1} \sum_{\ell=1}^{n} (-1)^\ell\frac{\ell}{2n}\binom{2n}{n+\ell} \ell^p \\
&= \frac{2^{2p}}{n} \sum_{\ell=1}^{n} (-1)^\ell \binom{2n}{n+\ell} \ell^{p+1}
\end{align*}
Similarly, 
\begin{align*}
S_{\text{odd}} &= 2^{p+1} \sum_{\ell=1}^n (-1)^\ell \frac{\ell+n}{2n} \binom{2n}{n+\ell} (2\ell-1)^p \\
&= 2^{2p+1} \sum_{\ell=1}^n (-1)^\ell \frac{\ell}{2n} \binom{2n}{n+\ell} \left(\ell-\frac12\right)^p + \frac{2^{2p+1}}{2} \sum_{\ell=1}^n (-1)^\ell \binom{2n}{n+\ell} \left(\ell-\frac12\right)^p\\
&=: S_1 + S_2
\end{align*}
We claim that $S_1, S_2$ and $S_{\text{even}}$ have the same sign. We first show that $S_2$ and $S_{\text{even}}$ have the same sign.

Applying \cite[Lemma 3.3]{LWW20} with the binomial expansion
\[
\left(\ell-\frac12\right)^p = \sum_{q=0}^p \binom{p}{q}(-1)^{p-q}\ell^{p-q}\frac{1}{2^q},
\]
we see the summation in $S_2$ contains only the terms corresponding to odd $q$, and the terms are
\begin{multline*}
\binom{p}{q}\frac{(-1)^{p-q}}{2^q} 2^{2n-(p-q)}\frac{(p-q)!}{\pi} \left(\sin\frac{\pi(p-q)}{2}\right) \int_0^\infty \frac{\sin^{2n}t}{t^{p-q+1}} dt \\= -\frac{2^{2n-p}p!\pi}{q!}\left(\sin\frac{\pi(p-q)}{2}\right) \int_0^\infty \frac{\sin^{2n}t}{t^{p-q+1}} dt.
\end{multline*}
The first term $q=1$ has the same sign as $S_{\text{even}}$ and the terms alternate in signs (for odd $q$). When $p=2$, the only odd term is $q=1$ and we know that $S_2$ and $S_{\text{even}}$ have the same sign. For $p\geq 4$, it suffices to show that
\[
a_q = \frac{1}{q!}\int_0^\infty \frac{\sin^{2n} t}{t^{p-q+1}}dt
\]
is decreasing in $q=1,2,\dots,p-1$. Since
\[
a_q - a_{q+1} = \frac{1}{q!}\int_0^\infty \frac{\sin^{2n}t}{t^{p-q+1}}\left(1-\frac{t}{q+1}\right)dt,
\]
it is equivalent to showing that
\[
\sum_{m=0}^\infty \int_{m\pi}^{(m+1)\pi}\frac{\sin^{2n}t}{t^{p-q+1}}\left(1-\frac{t}{q+1}\right)dt > 0.
\]
It is clear that the integral is dominated by the values on $[(m+1/2-\delta)\pi,(m+1/2+\delta)\pi]$ for large $n$ with an overall additive value of $O_p(\sin^n((\frac12-\delta)\pi)) = O_p(\exp(-c\delta n))$ for some absolute constant $c$, hence it suffices to show that
\[
\sum_{m=0}^\infty \int_{(m+\frac12-\delta)\pi}^{(m+\frac12+\delta)\pi}\frac{\sin^{2n}t}{t^{p-q+1}}\left(1-\frac{t}{q+1}\right)dt > c'
\]
for some absolute constant $c' > 0$.

Splitting the sum into $m\leq m_0$ and $m > m_0$ such that $t\leq q+1$ for $[(m+\frac 12-\delta)\pi,(m+\frac 12+\delta)\pi]$ for all $m\leq m_0$. It is then clear that for all $m \geq m_0 + 2$ we have $t > q+1$ on $[(m+\frac 12-\delta)\pi,(m+\frac 12+\delta)\pi]$. The only interval in which the integrand may change sign is when $m = m_0+1$ and one can verify that when $\delta\leq 0.1$, this can only happen for $m_0\geq 10$. When such an interval exists, it is easy to see that (since the integrand on the left-hand side is much larger)
\[
\int_{(\frac32+\delta)\pi}^{(\frac32-\delta)\pi}\frac{\sin^{2n}t}{t^{p-q+1}}\left(1-\frac{t}{q+1}\right)dt > \int_{(m_0+1+\frac12-\delta)\pi}^{(m_0+1+\frac12+\delta)\pi}\frac{\sin^{2n}t}{t^{p-q+1}}\left|1-\frac{t}{q+1}\right|dt.
\]
Now it suffices to show that (where $m'=m_0+2$ or $m_0+1$ depending on whether the special interval above exists)
\[
\int_{(\frac12-\delta)\pi}^{(\frac12+\delta)\pi}\frac{\sin^{2n}t}{t^{p-q+1}}\left(1-\frac{t}{q+1}\right)dt > \sum_{m=m'}^\infty \int_{(m+\frac12-\delta)\pi}^{(m+\frac12+\delta)\pi}\frac{\sin^{2n}t}{t^{p-q+1}}\left|1-\frac{t}{q+1}\right|dt + c'.
\]
Let
\[
I = \int_{(\frac12-\delta)\pi}^{(\frac12+\delta)\pi} \sin^{2n} t\ dt.
\]
Then $I = \Theta(1/\sqrt n)$. The integral on the right-hand side can be easily upper-bounded by
\[
\frac{I}{q+1} \sum_{m=m'}^\infty \frac{1}{((m+\frac12-\delta)\pi)^{p-q}} 
\leq \frac{I}{q+1} \sum_{m=1}^\infty \frac{1}{((m+\frac12-\delta)\pi)^{p-q}} 
 = \frac{I}{q+1}\zeta\left(p-q,\frac{3}{2}-\delta\right)
\]
and the integral on the left-hand side is lower-bounded by
\[
I\frac{1}{((\frac12+\delta)\pi)^{p-q+1}}\left(1-\frac{(\frac12+\delta)\pi}{q+1}\right).
\]
We would need
\begin{equation}\label{eqn:aux_reduced}
\frac{1}{(\frac12+\delta)^{p-q+1}}\left(1-\frac{(\frac12+\delta)\pi}{q+1}\right) > \frac{\pi}{q+1}\zeta\left(p-q,\frac{3}{2}-\delta\right) + c'
\end{equation}
for some absolute constant $c'$. Taking $\delta=0$, it suffices to show that
\[
2^{p-q} > 5\pi \zeta(p-q,\frac{3}{2})
\]
which hold for $p-q\geq 3$. Since both sides of \eqref{eqn:aux_reduced} are continuous in $\delta$, \eqref{eqn:aux_reduced} holds for some small $\delta > 0$. The only case left is when $p-q=2$. Since $p\geq 4$, $q\geq 2$, when $\delta=0$, it suffices to show that 
\[
4\left(1-\frac{\pi}{6}\right) > \frac{\pi}{3}\zeta\left(2,\frac{3}{2}\right),
\]
which is true, hence \eqref{eqn:aux_reduced} also holds when $p-q=2$ for some small $\delta>0$. This completes the proof that $S_2$ and $S_{\text{even}}$ are of the same sign.


A similar argument shows that $S_1$ and $S_{\text{even}}$ have the same sign. Thus by Lemma 3.3 and the argument in Lemma 3.4 of~\cite{LWW20}, we have
\[
|\hat g(s)| = |S_{\text{even}} + S_{\text{odd}}|\geq |S_{\text{even}}| \gtrsim \frac{2^{d/2}}{d^{3/2}}.
\]

\end{proof}

\subsection{Additive approximation algorithms}

We now design streaming algorithms that achieve an additive $\eps$-approximation to the objective Eqn.~\eqref{eqn:ridgeF}. We also generalize these results to a slightly modified (sum of squares) objective in \Cref{apx:squaredDist}. 
We start with dimension $d=1$.

\begin{theorem}\label{thm:add_ub_1d}
There exists a one pass streaming algorithm for the point estimation
variant of the problem in the one dimensional regime, that achieves an
additive error of $\eps$, space of $O(\eps^{-1/2}\sqrt{\log
(1/\eps)})$ words, and that succeeds with constant probability per query.
\end{theorem}


Recall that for $d=1$, the objective simplifies to $F(q) =\tfrac{1}{n}\sum_{i=1}^n \max\{0,q-x_i\}$. We describe a {\em sketching} algorithm that produces a sketch of
size $O(1/\sqrt{\eps})$ that is able to answer point estimation
queries to this $F$. Later, we show how to adapt this algorithm to the streaming
setting.

Let $m=(1/\sqrt \eps)$ and consider two sets of $m$ points. First consider $Y_1,\dots, Y_{m}$ such that $Y_i$ is at position $i/m$. Moreover consider $m$ points $X_1,\dots,X_{m}$ such that $X_i$ is at position $x_{(i\cdot n)/m}$, where we assume that $x_i$'s are in a sorted order, i.e., $x_1 \leq \cdots \leq x_n$. Now sort these $2m$ points and name them $Z_1,\dots, Z_{2m}$. For each $i\leq 2m$ we store three numbers: i) $Z_i$ itself, ii) $s_i$, the sum of the distances of the points to the left of $Z_i$ to the point $Z_i$, and  iii) $c_i$, the number of points $x_i$ to the left of $Z_i$.

Given a query $q\in [0,1]$, we will find $i$ such that $Z_i \leq q
<Z_{i+1}$. We will return $s_i + c_i\cdot (q-Z_i)$. Clearly for the
the points that are to the left of $Z_i$ this distance is computed
correctly. The only points that are not computed in the sum are part
of the points in the interval $[Z_i,Z_{i+1}]$, but we know that there
are at most $n\sqrt \eps$ of them (by our choice of the $X_i$'s) and their
distance to $q$ is at most $(Z_{i+1} - Z_{i}) \leq \sqrt \eps$ (by our
choice of $Y_i$'s). Therefore they introduce an average error of at most $\eps$ as we require.

\subparagraph*{Streaming.} We adapt the above algorithm to the streaming setting as follows. We
keep a binary tree, where each node corresponds to an interval in
$[-1,1]$ (the domain of $x_i$). The root corresponds to the entire
interval $[-1,1]$, and the two children of a node/inverval are the 2
half correspondingly (applied recursively). Initially, we start with a
tree of height $\log_2 1/\sqrt{\eps}$, where the leaves correspond to
intervals of length precisely $\sqrt{\eps}$ (assuming it's a power of two, w.l.o.g.).

As we stream through the points $x_i$, we add the information about
the point $x_i$ to the leaf corresponding to the interval containing $x_i$. In
particular, each node $v$, with associated interval $I_v$, keeps a
count of points $c_v$, as well as $s_v$ which is the sum, over of the
points accounted in $c_v$, of their distance to the right border of
the interval.

We may also expand this leaf $v$ to add its two children ($v$ ceases
to be a leaf). The leaf is split when $c_v$ reaches value
$\sqrt{\eps}n$. The new children start with their counters equal to
0. One exception is that if the depth of the node is more than $3\log
(1/\eps)$ (the interval's diameter is $<\eps$), in which case we don't do
the expansion.

To answer a query $q\in [-1,1]$, we sum up, over all nodes $v$
(internal nodes and leaves) whose interval $I_v$ is entirely to the left of
$q$, the quantity $s_v+c_v\cdot (q-I_v)$, where $q-I_v$ is the
distance from $q$ to (the rightmost endpoint of) $I_v$. 

We now argue the correctness and space complexity of this algorithm. 
\begin{proof}
For correctness,
the output of a query $q$ accounts for all the data points entirely to
the left of the interval $I_{v(q)}$, where $v(q)$ is the unique leaf
$v(q)$ with $q\in I_{v(q)}$. Hence the error comes entirely from the
unaccounted points in $v(q)$ as well as the predecessors of $v(q)$. By
construction any (non-empty) predecessor has diameter $\sqrt{\eps}$
and contains less than $\sqrt{\eps}n$ points, or, alternatively, has
diameter less than $\eps$ (for expansion exception). Hence the error
is $\le\eps$ for each predecessor, and $O(\eps\log 1/\eps)$
overall. As usual, we can rescale $\eps$ to obtain error $\eps$ and
correspondingly larger space.

For space complexity, note that the space is proportional to the size
of the tree. The tree has size at most $O(1/\sqrt{\eps})$ since each
leaf is either one of the original $2/\sqrt{\eps}$ one or has a parent
node with $\sqrt{\eps}n$ points.
\end{proof}

Now we study dimension $d = 2$. We now develop a streaming algorithm for sketching a set of points in the 2D
plane such that given any query (affine) line in the plane, one can
approximate the cost. 
To simplify the ensuing notation, we denote the
set of points by $p_1=(x_1,y_1),\dots,p_n=(x_n,y_n) \in
[0,1]\times [0,1]$, and the query by a line $\{x: \theta^T x  =b\}$, which we denote by $L = 
(\theta,b)$. Recall the assumption that $\|\theta\|\leq 1$, we may assume that $\|\theta\|=1$. Our goal is equivalent to reporting the sum of distances
of the points on one side of $L$ to $L$. Henceforth
 we denote the distance from point $p$ to line $L$ as $\dist(p,L)$. 

\begin{theorem}\label{thm:add_ub_2d}
There is a streaming algorithm for the point estimation 
problem in two dimensions, that with constant probability,
achieves additive error $\eps$, with sketch size 
$O(\eps^{-4/5})$ words.
\end{theorem}

The following shows how to get an $O(1/\eps)$-size sketch
with an additive error of $O(\eps^{5/4})$, and at the end
we just replace $\eps'=\Theta(\eps^{4/5})$ to
prove the above theorem.

We use a quad-tree over the unit square $[0,1]\times[0,1]$, where each
node is associated with a number of points (each point is associated
with exactly one node of the quad-tree). Thus each node $v$ contains a
counter $c_v$ for the number of associated points, a randomly chosen
associated points (chosen using reservoir sampling), as well as a sketch
$S_v$ to be described later. Initially, the quad-tree is of depth
$\log (1/\sqrt \eps)$ and all counters/sketches are initialized to
zero. When we stream over a point $p_i$, we associate it with the
corresponding leaf of the quad-tree (process defined later), unless
the counter $c_v$ is already $\eps\cdot n$ and the depth is at least
$2\log(1/\eps)$ (i.e., the side length is at least $\eps^2$). In that
case the leaf $v$ is expanded by adding its 4 children, which become
new leaves (with counters initialized to 0).

When we associate a point with a node $v$, we
increment $c_v$ and update the sketch $S_v$ on the associated
points. The sketch $S_v$ for the associated points, say termed $P_v$,
allows us to compute, for any query line $L=(\theta, b)$, the sum
$\sum_{p\in P_v} (b - \theta^Tp)$. The sum $\sum_{p\in P_v} \theta^T p$ can be computed in a
streaming fashion using the sketch from~\cite{pass-glm}. In
particular, the sketch actually consists of two counters: $X_v$, the sum of
the $x$ coordinates, and $Y_v$, the sum of the $y$ coordinates.


\subparagraph*{Query algorithm.} Given a query line $L$, we distinguish
contribution from points in two types of quad-tree nodes: nodes that do not intersect
the line and those that do. For the first kind, we can just use the
sketch $S_v$ to estimate the distance to the line, without incurring
any error. More precisely we have that 
$\sum_{p\in P_v} \dist(p,L) = \sum_{p\in P_v} (b - \langle p,\theta\rangle )
= c_v b - \langle (X_v,Y_v) , \theta \rangle$.
Note that this is
included in the final sum iff the entire node lies in the halfplane $\langle x,\theta\rangle\leq b$.

For the second kind of nodes, we estimate their contribution as
follows. For each non-empty node $v$, with the random sample $r_v$, we add to
the final sum the quantity $(1/n)\cdot c_v\cdot
\max\{0,b-\theta^Tr_v\}$. 

We now analyze this procedure.
\begin{proof}
It is clear that the total space usage of the
algorithm is at most order of the size of the quad-tree. The size of
the tree is bounded by $O(1/\eps)$ as follows. First, there at most
$O(1/\eps)$ of the original nodes. Second, each new children created
has the property that its parent got associated with $\eps n$ points,
hence at most $4/\eps$ such children can be ever created.

We now analyze the error of the sketching
algorithm. For the nodes that do not cross the line $L$, the
distances of their points are computed exactly. So we only need to
argue about the crossing nodes.  First of all, note that we can ignore
all leaves with more than $\eps n$ points at them as their diameter is
less than $2\eps^2$. 

Let $\mathcal{C}$ be the set of leaves that cross the query line and
have diameter at least $2\eps^2$ (and hence less than $\eps n$
associated points). It is immediate to check that, in expectation, our
estimator outputs the correct value; in particular for $P_v$ the set
of points associated to a node $v$:
\[
\E\left[\sum_{v\in \mathcal{C}} c_v \cdot \max\{0,\dist(r_v,L)\}\right] = \sum_{v\in \mathcal{C}}\sum_{p\in P_v} \frac{1}{c_v} c_v \cdot \max\{0,\dist(p,L)\} = \sum_{v\in \mathcal{C}}\sum_{p\in P_v}\max\{0,\dist(p,L)\}. 
\]
Thus we only need to argue that it concentrates closely around its
expectations, with constant probability. Let us compute the
variance. The point in each (non-empty) node is chosen independently at
random. Thus we can sum up the variances of each node. Consider a
node $v$ with side length $\ell\ge \eps^2$. Then we have that
\[
 \Var\left[ c_v \cdot \max\{0,\dist(r_v,L)\}\right]  \leq \sum_{i=1}^{c_v} \frac{1}{c_v}\cdot (c_v\cdot\ell)^2 \leq c_v^2\ell^2 \leq (n\eps)^2\ell^2.
\]
Now note that any line can intersect only $(1/\ell)$ nodes with side
length $\ell$, and thus the total variance of all nodes with side
length $\ell$ is at most $(n\eps)^2\ell$. Hence, the total variance
over all nodes (over all levels) is at most
$O(n^2\eps^2\sqrt\eps)$. Overall, the standard deviation is at most $
O(n\eps^{5/4})$. By Chebyshev's bound, the reported answer has an
additive error of $O(n\eps^{5/4})$ with constant probability.

As stated before, replacing $\eps' = \Theta(\eps^{4/5})$ we get that the
algorithm is providing an additive $\eps'$ approximation using space
$\tilde O((\eps')^{-4/5})$, completing the proof of the result.
\end{proof}

\section{Optimization}

In this section we consider the problem of finding the (approximate) optima for the SVM objective in the streaming setting. First, we show that a streaming solution for the point estimation problem immediately leads to a solution for the optimization problem, with only a $O(d\log 1/\eps)$ loss in space complexity. Second, we give lower bounds for the optimization problem, showing a (different) polynomial dependence on $1/\eps$ is still required for dimension $d > 1$.

As before, we consider the SVM optimization problem in which the bias is regularized:
$\min_{\theta, b} F_\lambda(\theta, b)$, where $F_\lambda$ is as defined in \eqref{eqn:F_lambda}. Without loss of generality, we assume that the inputs are contained in a ball of radius $1$, i.e., $\|x_i\| \leq 1$, and that $y_i \in \{-1,+1\}$.

Recall that \cite{shalev-shwartz2007pegasos} show that $O(1/(\lambda\eps))$ random samples $(x_i,y_i)$ are enough for computing an $\eps$-approximate optimum (by running SGD). This can be seen as a streaming algorithm with space complexity $O(d/(\lambda\eps))$. We show that, given an efficient streaming algorithm for point estimation, we can solve the optimization problem with only a minor blowup.

\begin{theorem}
\label{thm:pe2opt}
    Suppose there is a streaming algorithm that, after seeing data $\{(x_i,y_i)\}_{i=1}^n$, where $\|x_i\| \leq 1$, can produce a sketch of size $s$ that, given any $(\theta,b)$ such that $\|(\theta, b)\| \leq \sqrt{2/\lambda}$, is able to output $\hat F(\theta, b)$ such that $|\hat F(\theta,b) - F(\theta, b)|\leq \eps$ with probability at least $0.9$. Then there is also a streaming algorithm that, under the same input, will output $(\hat \theta, \hat b)$ with $|F_\lambda(\hat\theta,\hat b)-F_\lambda(\theta^*_\lambda,b^*_\lambda)|\le \eps$ with probability at least $0.9$, while using space $O(s \cdot d\log d/(\lambda \eps))$.
\end{theorem}
\begin{proof}
For notational simplicity, let $w = (\theta, b) \in \R^{d+1}$ and
$w^*_\lambda := \argmin  F_\lambda(w)$. Define $F(w) =
\frac{1}{n}\sum_{i=1}^n \max\{0, 1-w^T z_i\}$, where $z_i = y_i(x_i,
1)$.

First, note that we only need point estimation for $w$ of norm
$\|w\| \le R$, where $R=\sqrt{\frac{2}{\lambda}}$. Specifically, since $1 = F_\lambda(0) \geq F_\lambda(w^*_\lambda) - F_\lambda(0)$, we have by Lemma~\ref{lem:str_conv} that
\[ 
    \|w_\lambda^*\| \leq \sqrt{\frac{2}{\lambda}} = R.
\] 
Let $T := \delta \mathbb{Z}^{d+1} \cap B(0, R)$ be a grid of cells with side length $\delta$ that is also contained in a ball of radius $R$. We can upper bound the size of $T$ by $|T| = O((2R/\delta)^{d+1})$. 

\textbf{Claim:} If $\|w-w^*_\lambda\|_{\infty} \leq \delta := \frac{\eps}{2\sqrt{d}}$ and $\|x_i\| \leq 1$ then $F(w) - F(w^*_\lambda) \leq \eps$.

To see this, note that $\|w-w_\lambda^*\|_{\infty} \leq \delta$ implies that $\|w-w_\lambda^*\|_2 \leq \sqrt{d}\delta = \eps/2$. Since we assume $\|x_i\| \leq 1$, we also have $\|z_i\| \leq 2$. Hence
\begin{align*}
    F(w) - F(w^*_\lambda) &= \frac{1}{n}\sum_{i=1}^n \left(\max\{0, 1- w^T z_i\} - \max\{0, 1- (w_\lambda^*)^T z_i\}\right)\\
    &\leq \frac1n \sum_{i=1}^n \|w-w_\lambda^*\|\, \|z_i\| \\
    &\leq \frac{\eps}{2} \cdot 2 = \eps
\end{align*}
as desired.

Suppose the sketch uses space $s$ and has failure probability $0.1$. To boost this probability, given $(\theta, b)$, we can run $k = O(d\log \frac{1}{\lambda \eps}) = O(\log |T|)$ sketches in parallel, then output the median value. The median is within an additive error of $\pm\eps$ except with failure probability $O(2^{-k})$. Hence, we can take the union bound over all $T$, so that with probability at least $0.9$ the resulting sketch succeeds for \textit{all} $w \in T$, while using $O(sd\log\frac{d}{\lambda \eps})$ space. Moreover, by the first claim, some $w \in T$ satisfies $F_\lambda(w) - F_\lambda(w^*_\lambda) \le \eps$. Hence, we can iterate over all $w \in T$ and output the parameter that achieves the lowest value of $F_\lambda(w)$, proving the theorem.

\end{proof}

Recall that our point estimation results assume $\|(\theta, b)\| \le 1$, which can be adapted to $\|(\theta, b)\| \le R$ by replacing $\epsilon$ with $\epsilon/R$. Letting $R = \sqrt{2/\lambda}$, the above theorem implies that we get an optimization algorithm for $d=1$ that uses $\tilde O(\eps^{-1/2} \lambda^{-1/4})$ space, and an optimization algorithm for $d=2$ that uses $\tilde O(\eps^{-4/5} \lambda^{-2/5})$ space. Note that this has a polynomially better dependence on both $\eps$ and $\lambda$ relative to the $O(1/\lambda\eps)$ bound that we get from SGD.





\subparagraph*{Lower bounds.}
We start with the high-dimensional case. Suppose there exists a sketch such that, given a stream of inputs $\{(x_i, y_i)\}_{i=1}^n$, $\|x_i\|_2 \leq 1$, with probability at least $0.9$ outputs some $(\hat{\theta}, \hat{b})$ such that 
$F_\lambda(\hat{\theta}, \hat{b}) \leq F_\lambda(\theta^*, b^*) + \eps$,
where $F_\lambda(\theta, b) := \frac{\lambda}{2}(\|\theta\|^2 + b^2) + \frac{1}{n}\sum_{i=1}^n \max\{0, 1-y_i(\theta^T x_i + b)\}$. 
For now, suppose $\lambda = 10^{-4}$ and $d = O(\log n)$. Later, we show a similar lower bound for low dimensions when $\lambda = \Theta(1/n)$.

\begin{theorem}\label{thm:main_thm}
Such a sketch requires $\Omega(\eps^{-1/2})$ words of space.
\end{theorem} 

\textit{Proof idea:}
We will reduce from the \textsc{Indexing} problem in the one-way communication model. 
At a high level, we will argue that we can query whether a point exists at a given location on the unit sphere by adding additional points with the property that the optimal parameters are determined entirely by the added points and the point being queried (if it exists). This yields a separation in the optimal parameters in these two cases, which we will argue (via strong convexity) is distinguishable using such a sketch. 

Specifically, suppose Alice is given a bit string $b \in \{0, 1\}^{n/2}$ which she wants to encode. Let $T = \{v_1, \ldots, v_n\}$ be a subset of the unit sphere in $d$ dimensions satisfying $\forall v_1 \neq v_2 \in T$, $v_i^T v_j < 1 - 10\delta$, where $\delta$ will be specified later, with $|T| = n$. For $\delta = \frac{1}{100}$, such a set exists for $d = \poly(\log n)$. If $b_i = 0$, Alice adds $(v_i, -1)$ to the sketch $S$; otherwise, if $b_i = 1$, Alice adds $(v_{n/2 + i}, -1)$ to $S$. 
Alice then sends the sketch to Bob, who decodes $b_q$ for $q \in [\frac{n}{2}]$ by querying whether a point exists at location $x_q$. In particular, Bob adds $\frac{n}{4}$ copies of $(x_\alpha, -1)$, and $\frac{n}{4}$ copies of $(x_\beta, +1)$, where $x_\alpha := (1-\delta)x_q$ and $x_\beta := (1+\delta)x_q$. Let $(\hat{\theta}, \hat{b})$ be the output of the sketch after doing so. Define $\theta_0 := \frac{2\lambda(1+\delta) + \delta}{2\lambda(1+(1+\delta)^2)}$ and $\theta_1 := \frac{2\lambda(1+\delta) + \delta(1+\frac{1}{n})}{2\lambda(1+(1+\delta)^2)}$. 
If $\|\hat{\theta} - \theta_0\| < \|\hat{\theta} - \theta_1\|$ then Bob outputs ``$b_i = 0$''; otherwise, Bob outputs ``$b_i = 1$''.


\subsection{Preliminaries for the lower bound}\label{apx:subsec:kkt}
We introduce some standard facts about the SVM objective. We can rewrite the SVM objective as a constrained optimization problem:
\begin{equation}
    \min_{\theta, b} \frac{\lambda}{2}\|(\theta, b)\|_2^2 + \frac{1}{n}\sum_{i=1}^n \gamma_i
\end{equation}
subject to $\gamma_i \geq 0$ and $\gamma_i \geq 1-y_i(\theta^T x_i + b)$.

The corresponding Lagrangian is then:
\begin{equation}
    \min_{\theta, b, \gamma, \eta, \alpha} \mathcal{L} := \frac{\lambda}{2}\|(\theta, b)\|_2^2 + \frac{1}{n}\sum_{i=1}^n \gamma_i - \sum_{i=1}^n \gamma_i \eta_i + \sum_{i=1}^n \nu_i [1-\gamma_i-y_i(\theta^T x_i + b)]    
\end{equation}
subject to:
$\gamma_i \geq 0$, $\gamma_i \geq 1-y_i(\theta^T x_i + b)$, $\eta_i \geq 0$, and $\nu_i \geq 0$. 

The KKT conditions are:
\begin{gather}
    \eta_i \gamma_i = 0 \\
    \nu_i (1-\gamma_i-y_i(\theta^T x_i + b)) = 0 \\
    \theta^* = \frac{1}{\lambda}\sum_{i=1}^n \nu_i y_i x_i \label{eqn:kkt:theta} \\ 
    b^* = \frac{1}{\lambda}\sum_{i=1}^n \nu_i y_i \\
    \eta_i = \frac{1}{n} - \nu_i \\
    \gamma_i = \max\{0, 1-y_i(\theta^T x_i + b)\} \\
    \eta_i \geq 0, \nu_i \geq 0
\end{gather}

\subsubsection{Notation}
We will characterize the optimal solutions for the two scenarios (there exists a point at $x_q$ or there does not) for the case that the dimension is $d=1$. We will later show how this provides results for $d > 1$. We will specify $\delta$ later, but will always maintain the relation $\lambda = \delta^2$. 

Let $n=\frac{1}{20\sqrt{\eps}}$, and define a set $S_0$ of $n$ points as follows:
\begin{itemize}
    \item $\frac{n}{4}$ points are $x_\alpha := 1-\delta$ with $y_\alpha := -1$.
    \item $\frac{n}{4}$ points are $x_\beta := 1+\delta$ with $y_\beta := +1$.
    \item The remaining $\frac{n}{2}$ points are not support vectors and are otherwise arbitrary; i.e. are such that 
    $1-y((\theta^*_0)^T x + b^*_0) < 0$, where
    \begin{equation}
        (\theta^*_0, b^*_0) = \arg\min_{(\theta, b)} F_\lambda^{(0)}(\theta, b) := \frac{\lambda}{2}(\|\theta\|^2 + b^2) + \frac{1}{n}\sum_{(x, y) \in S_0} \max\{0, 1-y(\theta^T x + b)\}
    \end{equation}
    Observe that such points exist as long as $\theta^*_0 \neq 0$, which is easy to show.
\end{itemize}

Define $\gamma_\alpha^{(0)} := \max\{0, 1-y_\alpha((\theta^*_0)^T x_\alpha + b^*_0)\} = \max\{0, 1+\theta^*_0(1-\delta) + b^*_0\}$ and $\gamma_\beta^{(0)} := \max\{0, 1-y_\beta((\theta^*_0)^T x_\beta + b^*_0)\} = \max\{0, 1-\theta^*_0(1+\delta) - b^*_0\}$. 

Similarly, define $S_1$ to be the set of $n$ points that is exactly the same as $S_0$, except that instead of $\frac{n}{2}$ points that are not support vectors, there are $\frac{n}{2}-1$ points that are not support vectors and one additional point at $x_q := 1$ with $y_q := -1$. Similarly define $F_\lambda^{(1)}$, $(\theta^*_1, b^*_1)$, $\gamma_\alpha^{(1)}$, and $\gamma_\beta^{(1)}$ to be the analogous quantities where $S_0$ is replaced with $S_1$, and let $\gamma_q^{(1)} := \max\{0, 1-y_\beta((\theta^*_1)^T x_q + b^*_1)\} = \max\{0, 1+\theta^*_1 + b^*_1\}$. We will also sometimes use the notation $\gamma_\alpha(\theta, b) := \max\{0, 1-y_\alpha(\theta^T x_\alpha + b)\}$, and similarly for $\gamma_\beta(\theta, b)$. 

\subsubsection{Lemmas}
Next, we show some properties of the optimal solutions in each of the two cases. 

\begin{lemma}\label{lem:optlb1}
Suppose $\delta < \frac{1}{100}$ and $n \geq \frac{1}{\delta^2}$. Then we have that $\gamma_\alpha^{(i)} > 0$ and $\gamma_\beta^{(i)} = 0$ for $i \in \{0, 1\}$.
\end{lemma}

\begin{proof}
To see that $\gamma_\alpha^{(i)} > 0$ and $\gamma_\beta^{(i)} = 0$ for $i \in \{0, 1\}$, we will rule out the other possibilities. 

\noindent\textbf{Claim:} Suppose $\delta < \frac{1}{7}$ and $n \geq \frac{1}{\delta^2}$. Then we cannot have $\gamma_\beta^{(i)} > 0$.

Suppose, for the sake of contradiction, that $\gamma_\beta^{(i)} > 0$, and let $(\theta_i^*, b_i^*)$ be the corresponding optimal parameters. First consider $i=0$. Then since $\gamma_\beta^{(0)} > 0$, by the KKT conditions we have that the dual variables corresponding to each of the $\frac{n}{4}$ copies of $x_\beta$, call them $\nu_\beta$, satisfy $\nu_\beta = \frac{1}{n}$. Moreover, note that this is the maximal value possible for the dual variables. Additionally, for points that are not support vectors, the dual variables are $\nu_i = 0$. The KKT conditions then imply:
\begin{equation}
    \theta^*_0 = \frac{1}{\lambda}\sum_{i=1}^{n} \nu_i x_i y_i \geq \frac{1}{\lambda}\left(\frac{n}{4} \frac{1}{n} (1+\delta) - \frac{n}{4} \frac{1}{n}(1-\delta) \right) = \frac{\delta}{2\lambda} = \frac{1}{2\delta}
\end{equation}
and $b^*_0 = \frac{1}{\lambda} \sum_{i=1}^{n} \nu_i y_i \geq 0$.
But this implies that $1-\theta_0^*(1+\delta)-b_0^* \leq 1-\frac{1}{2\delta}(1+\delta) = \frac{1}{2}(1-\frac{1}{\delta}) < 0$ (since $\delta < 1$), which contradicts that $\gamma_\beta^{(0)} > 0$.

Next, consider when $i=1$. This time, we also have $x_q$ as a support vector. Again using the fact that the corresponding dual variable is at most $\nu_q \leq \frac{1}{n}$, we have:
\begin{equation}
    \theta^*_1 = \frac{1}{\lambda}\sum_{i=1}^{n} \nu_i x_i y_i \geq \frac{1}{\lambda}\left(\frac{\delta}{2}-\frac{1}{n}  \right) \geq \frac{1}{\lambda}\left(\frac{\delta}{2}-\delta^2  \right) \geq \frac{1}{\lambda}\left(\frac{\delta}{2}-\frac{\delta}{4}  \right) = \frac{1}{4\delta}
\end{equation}
and $b^*_1 = \frac{1}{\lambda} \sum_{i=1}^{n} \nu_i y_i \geq -\frac{1}{\lambda n}$.
But this implies that $1-\theta_1^*(1+\delta)-b^*_1 \leq 1-\frac{1}{4\delta}(1+\delta)+\frac{1}{\delta^2 n} \leq \frac{1}{4}(3-\frac{1}{\delta}) + 1 < 0$, contradicting that $\gamma_\alpha^{(1)} > 0$ and proving the claim.

\textbf{Claim:} It cannot hold that $\gamma_\alpha^{(i)} = \gamma_\beta^{(i)} = 0$ for either $i \in \{0, 1\}$.

First consider $i=0$. We will show that in this case $\theta^*_0 = \frac{1}{\delta}$ and $b^*_0 = -\frac{1}{\delta}$. By assumption, we know that $1 + \theta^* + b^* - \theta^* \delta \leq 0$ and $1 - (\theta^* + b^*) - \theta^* \delta \leq 0$.
Summing these implies that 
\begin{equation}
    1 - \theta^*_0 \delta \leq 0 \Rightarrow \theta^*_0 \geq \frac{1}{\delta}
\end{equation}
Combining this last result with $\gamma_\alpha^{(0)} = 0$:
\begin{equation}
    0 \geq 1 + \theta^*_0(1-\delta) + b^*_0 \geq 1 + \frac{1}{\delta}(1-\delta) + b^*_0 = \frac{1}{\delta} + b^*_0 \Rightarrow b^*_0 \leq -\frac{1}{\delta}
\end{equation}
We have that $\theta = \frac{1}{\delta}$ and $b = -\frac{1}{\delta}$ satisfy $\gamma_\alpha(\theta, b) = \gamma_\beta(\theta, b) = 0$. By these last two equations, they are also clearly the \textit{smallest norm} values satisfying these; hence they indeed minimize the overall optimization problem. But then $F_0(\theta, b) = \frac{\lambda}{2}(\frac{2}{\delta^2}) = 1$. But since we also have that $F_0(0, 0) = 1$, this cannot be the optimal solution, so we indeed cannot have $\gamma_\alpha^{(0)} = \gamma_\beta^{(0)} = 0$.

Now consider when $i=1$. By the same argument as before, we must again have that if $\gamma_\alpha^{(1)} = \gamma_\beta^{(1)} = 0$ then $\theta^*_1 \geq \frac{1}{\delta}$ and $b^*_1 \leq -\frac{1}{\delta}$. Observe also that $F_1(\theta, b) = F_0(\theta, b) + \frac{1}{n}\gamma_q(\theta, b) \geq F_0(\theta, b)$. Hence:
\begin{equation}
    F_1(\theta_1^*, b_1^*) \geq F_0(\theta_1^*, b_1^*) = \frac{\lambda}{2}(2\frac{1}{\delta^2}) = 1
\end{equation}
which again contradicts the optimality of $(\theta_1^*, b_1^*)$ since $F_1(0, 0) = 1$. The proof of the claim is complete.

Since under the assumption of the lemma we cannot have $\gamma_\alpha^{(i)} = \gamma_\beta^{(i)} = 0$, and we cannot have that $\gamma_\beta^{(i)} > 0$, we must have that $\gamma_\alpha^{(i)} > 0$ and $\gamma_\beta^{(i)} = 0$, concluding the proof.
\end{proof}

Next, we show what the optimal parameters are in this case.

\begin{lemma}\label{lem:optlb2}
Suppose that $\gamma_\alpha^{(0)} > 0$ and $\gamma_\beta^{(0)} = 0$. Then we have:
\begin{equation}
    \theta^*_0 = \frac{2\lambda(1+\delta) + \delta}{2\lambda(1+(1+\delta)^2)}
\end{equation}
and $b^*_0 = 1-(1+\delta)\theta^*_0$.

Similarly, if $\gamma_\alpha^{(1)} > 0$ and $\gamma_\beta^{(1)} = 0$. Then we have:
\begin{equation}
    \theta^*_1 = \frac{2\lambda(1+\delta) + \delta(1+\frac{1}{n})}{2\lambda(1+(1+\delta)^2)}
\end{equation}
and $b^*_1 = 1-(1+\delta)\theta^*_1$.
\end{lemma}
\begin{proof}
First, observe that for $i \in \{0, 1\}$ we not only have $\gamma_\beta^{(i)} = \max\{0, 1-(\theta^*_i + b^*_i) - \theta^*_i \delta\} = 0$, but also have that \begin{equation}
    1-(\theta^*_i + b^*_i) - \theta^*_i \delta = 0    
\end{equation}
To see this, suppose otherwise, that $1-(\theta^*_i + b^*_i) - \theta^*_i \delta < 0$. Then there exists some $0 < \theta' < \theta^*_i$ such that we still have $1-(\theta' + b^*_i) - \theta' \delta = 1-\theta'(1 + \delta) - b^*_i < 0$. Moreover, this can only decrease the regularization term, $\frac{\lambda}{2}\theta^2$, and can only decrease the $\gamma_\alpha^{(i)} = \max\{0, 1+\theta(1-\delta)+b\}$ term. In the case that $i=1$, this can also only decrease the $\gamma_q^{(1)} = \max\{0, 1+\theta + b\}$ term. Either way, this contradicts the optimality of $\theta^*_i$. Hence, $b^*_i = 1-\theta^*_i(1+\delta)$. 

For $i=0$, plugging this in and simplifying, we can rewrite the optimization problem as:
\begin{equation}
    \min_{\theta} \frac{\lambda}{2}(\theta^2 + (1-\theta(1+\delta))^2) + \frac{1}{2}(1 - \theta \delta)
\end{equation}
Differentiating with respect to $\theta$ and setting the expression equal to zero,
\begin{align*}
    0 &= \left.\frac{\partial}{\partial \theta} \left[\frac{\lambda}{2}(\theta^2 + (1-\theta(1+\delta))^2) + \frac{1}{2}(1 - \theta \delta)\right]\right|_{\theta=\theta_0^\ast} \\
    &= \frac{\lambda}{2}(2\theta_0^* - 2(1-\theta_0^*(1+\delta))(1+\delta)) - \frac{\delta}{2} \\
    &= \lambda\theta_0^*(1+ (1+\delta)^2) - \lambda(1+\delta) - \frac{\delta}{2} 
\end{align*}
from which we can solve for $\theta^*_0$ and obtain that
\[
\theta_0^* = \frac{2\lambda(1+\delta) + \delta}{2\lambda(1+(1+\delta)^2)}.
\]

Similarly, for $i=1$, we can rewrite the optimization problem as:
\begin{equation}
    \min_{\theta} \frac{\lambda}{2}(\theta^2 + (1-\theta(1+\delta))^2) + \frac{1}{4}(1+\theta + 1 - \theta(1+\delta) - \theta \delta) + \frac{1}{n}(1+\theta+1-\theta\delta)
\end{equation}
We can again differentiate with respect to $\theta$ and set the expression equal to zero, then solve for $\theta^*_1$, yielding the desired expression.
\end{proof}

We now adapt these results to the case that we care about: when $d > 1$ and there are points other than just $x_\alpha$, $x_\beta$, and $x_q$. 

First, note that when $d > 1$ and the only possible support vectors are $x_\alpha$, $x_\beta$, and $x_q$, $\theta^*_i \in \R^d$ is parallel to $x_q$. Hence, $\theta^*_i$ reduces to the one-dimensional case, and is simply projected onto $x_q$. For $d$ dimensions, we can thus replace what was previously $\theta_i^* \in \R$ with $(\theta_i^*)^T x_q \in \R$, where  $\theta_i^*, x_q \in \R^d$.

Recall that up to this point we have been assuming that the remaining $\frac{n}{2}$ or $\frac{n}{2} - 1$ points that were added by Alice are \textit{not} support vectors. We will now show that this is the case. 

Redefine $S_0$ to be the set of points including $\frac{n}{4}$ copies of $(x_\alpha, -1)$, $\frac{n}{4}$ copies of $(x_\beta, +1)$, and $\frac{n}{2}$ arbitrary points $\{v_i\}_{i=1}^{n/2}$, but this time with the requirement that $v_i^T x_q < 1-10\delta$, $\forall i \in [\frac{n}{2}]$. Similarly, let $S_1$ be the set of points including $\frac{n}{4}$ copies of $(x_\alpha, -1)$, $\frac{n}{4}$ copies of $(x_\beta, +1)$, one copy of $(x_q, -1)$ and $\frac{n}{2}-1$ arbitrary points $\{v_i\}_{i=1}^{n/2-1}$, all satisfying $v_i^T x_q < 1-10\delta$. 

Define
\begin{equation}
    F_\lambda^{(0)}(\theta, b; \lambda) := \frac{\lambda}{2}(\|\theta\|^2 + b^2) + \frac{1}{n}\sum_{(x, y) \in S_0} \max\{0, 1-y(\theta^T x + b)\}
\end{equation}
and
\begin{equation}
    F_\lambda^{(1)}(\theta, b; \lambda) := \frac{\lambda}{2}(\|\theta\|^2 + b^2) + \frac{1}{n}\sum_{(x, y) \in S_1} \max\{0, 1-y(\theta^T x + b)\}
\end{equation}

Finally, let $(\theta_i^*, b_i^*) = \argmin_{\theta, b} F_\lambda^{(i)}(\theta, b)$ for $i \in \{0, 1\}$.

\begin{lemma}\label{lem:optlb3}
Suppose $\lambda = \delta^2$, $\delta < \frac{1}{7}$ and $n \geq \frac{1}{\delta^2} = \frac{1}{\lambda}$. Then for both $S_0$ and $S_1$, none of the points $v$ satisfying $v^T x_q < 1-10\delta$ are support vectors. Moreover, we have
\begin{equation}
    \theta_0^* = \left(\frac{2\lambda(1+\delta) + \delta}{2\lambda(1+(1+\delta)^2)}\right)x_q
\end{equation}
\begin{equation}
    \theta_1^* = \left(\frac{2\lambda(1+\delta) + \delta(1+\frac{1}{n})}{2\lambda(1+(1+\delta)^2)}\right)x_q
\end{equation}
and $b_i^* = 1-\|\theta_i^*\|(1+\delta)$ for both $i \in \{0, 1\}$.
\end{lemma}
\begin{proof}
By the preceding discussion and lemmas, it suffices to show that for $(v, -1) \in S_i$ such that $v^T x_q < 1-10\delta$, $(v, -1)$ is not a support vector. This implies that $(\theta_i^*, b_i^*)$ does not depend on such $v$, so that the equations for the optimal parameters from before (adapted slightly for $d > 1$) indeed hold. For this, it suffices to show that for the $(\theta_i^*, b_i^*)$ described above, we have $1+((\theta^*_i)^T v + b^*_i) < 0$ when $v^T x_q < 1-10\delta$.


For the given $\delta$, $\lambda$, and $n$, one can easily verify that for both $i=0$ and $i=1$ we have $\|\theta_i^*\| \geq \frac{\delta}{5\lambda}$. Moreover, since $b_i^* = 1-(1+\delta)\|\theta_i^*\|$ and $(\theta_i^*)^T x_q = \|\theta_i^*\|$, we have:
\[
    1+(\theta_i^*)^T v + b_i^* \leq 1 + (1-10\delta)\|\theta_i^*\| + 1-(1+\delta)\|\theta_i^*\| 
    = 2 - 11\delta \|\theta_i^*\| \leq 2 - 11\delta \frac{\delta}{5\lambda} = 2 - \frac{11}{5} < 0
\]
proving the claim.
\end{proof}

\subsection{Proofs of Main Results}
We now use the preceding lemmas to complete the proof of Theorem~\ref{thm:main_thm}.
\begin{proof}
Let $n = \frac{1}{20\sqrt{\eps}}$, and let $x_q$ be the point being queried by Bob. If $x_q$ was added to the set by Alice, then $i := b_q = 0$; otherwise, $i := b_q = 1$. Let $(\hat{\theta}, \hat{b})$ be the output of the sketch. Then using the guarantee of the sketch and of strong convexity, we have that $\|(\hat{\theta}, \hat{b}) - (\theta^*_i, b^*_i)\| \leq \sqrt{\frac{2\eps}{\lambda}}$. Hence, to distinguish these two scenarios, it suffices to show that $\|(\theta^*_0, b^*_0) - (\theta^*_1, b^*_1)\| > 2\sqrt{\frac{2\eps}{\lambda}}$. Indeed,
\[
    \|(\theta^*_0, b^*_0) - (\theta^*_1, b^*_1)\| \geq \|\theta^*_0 - \theta^*_1\| = \frac{1}{n}\frac{\delta}{2\lambda(1+(1+\delta)^2)}
     \geq \frac{\delta}{5\lambda n} = \frac{20\delta\sqrt{\eps}}{5\lambda} = 4\sqrt{\frac{\eps}{\lambda}} > 2\sqrt{\frac{2\eps}{\lambda}},
\]
proving the claim.
\end{proof}

Now we extend this result to the low-dimensional case. When $d$ is a constant, the above lower bound cannot be directly applied because there does not exist a set $T$ of size $|T| = n$ such that $\forall v \neq v' \in T$, $v^T v' < 1-10\delta$ when $\delta$ is a constant. However, we can adapt the lower bound to the low dimensional setting if we let $\delta$ be sub-constant. Since we always maintain the relationship that $\lambda = \delta^2$, this means that $\lambda$ also scales with $\frac{1}{n}$. Note that $\lambda = \Theta(\frac{1}{n})$ is often used in practice.

\begin{theorem}\label{thm:opt_lb:thm2}
    For $d = 2$ and $\lambda = \Theta(\frac{1}{n^2})$, a sketch as defined above requires space $\Omega(\eps^{-1/4})$. For $d \geq 3$ and $\lambda = \Theta(\frac{1}{n})$, such a sketch requires space $\Omega(\eps^{-1/2})$. 
\end{theorem}

\begin{proof}
First, note that if $n = O(\delta^{-(d-1)})$ for $d > 1$, then there exists a subset $T$ of the unit sphere in $\mathbb{R}^d$ satisfying $\forall v \neq v' \in T$, $v^T v' \leq 1-10\delta$. Alice will use such a set to encode bits.

In the high dimensional lower bound construction, the constraints we had on the points were $\delta < \frac{1}{7}$ and $n \geq \frac{1}{\delta^2} = \frac{1}{\lambda}$. Moreover, for the final step of the analysis, we needed that $\frac{\delta}{5\lambda n} = \frac{1}{5\delta n} > 2\sqrt{\frac{2\eps}{\lambda}} = 2\sqrt{2\eps}\frac{1}{\delta}$, which is always satisfied for $n \leq \frac{1}{20\sqrt{\eps}}$ (regardless of $\lambda$ and $\delta$, as long as $\lambda = \delta^2$). Additionally, for the construction we must also have $n \leq \Theta(\delta^{-(d-1)})$. Hence, for $d = 3$, by letting $\lambda = \frac{1}{n} = \Theta(\eps^{1/2})$, all of these constraints are satisfied, so that we again get a $\frac{1}{\sqrt{\eps}}$ lower bound. For $d=2$ and $\lambda = \Theta(\eps^{1/2})$, this also implies that we can encode $\Theta(\delta^{-(d-1)}) = \Theta(\eps^{-1/4})$ bits.
\end{proof}

\newpage 

{\bf Acknowledgments:} Alexandr Andoni was supported in part by Simons Foundation (\#491119) and NSF (CCF-1617955, CCF- 1740833). Yi Li was supported in part by Singapore Ministry of Education (AcRF) Tier 2 grant MOE2018-T2-1-013. David P. Woodruff was supported by the National Science Foundation under Grant No. CCF-1815840.

\bibliography{camera_ready}

\appendix
\newpage

\section{Proof of Lemma~\ref{lem:str_conv}}\label{sec:str_conv_pf}

\begin{proof}
Note that $F_\lambda$ is $\lambda$-strongly convex. Let $w = (\theta, b)$ for notational simplicity. Then:
\begin{equation}
    F_\lambda(\hat{w}) \geq F_\lambda(w^*) + \nabla F_\lambda(w^*)^T(\hat{w}-w^*) + \frac{\lambda}{2}\|\hat{w}-w^*\|_2^2 = F_\lambda(w^*) + \frac{\lambda}{2}\|\hat{w}-w^*\|_2^2
\end{equation}
so that
\begin{equation}
    \|\hat{w}-w^*\|_2 \leq \sqrt{\frac{2}{\lambda} (F_\lambda(\hat{w}) - F_\lambda(w^*))} \leq \sqrt{\frac{2\eps}{\lambda}}
\end{equation}
as claimed.
\end{proof}


\section{Multiplicative $(1+\eps)$ sketching algorithm when $d=1$}\label{apx:pe_ub_sketch}
\subparagraph*{Sketching algorithm.} Assume that the points are sorted in increasing order, i.e., $x_1\leq \cdots \leq x_n$. In the sketch, the algorithm will store for each $i \in I =\{\ceil{(1+\eps)^j} | 0\leq j\leq\log_{(1+\eps)} n\}$ two values: the first is the position of the $i$-th point $x_i$, 
and the second is the sum of the distances between $x_i$ and the points to the left of $x_i$, i.e., $S_i = \sum_{j\leq i} (x_i - x_j)$. 

\begin{observation}
The size of the sketch is $O(\frac{\log n}{\eps}\cdot (\log n + \log W))$.
\end{observation}

\subparagraph*{Query algorithm.} Given the query $q\in \mathbb{R}$, report $T = S_j + j\cdot (q-x_j)$, where $j$ is the largest number in $I$ such that $x_j \leq q$.

\subparagraph*{Analysis.} Note that our estimate is equal to $T = \sum_{i\leq j} (x_j - x_i + q - x_j)$. Thus for all $i\leq j$, the distance $q-x_i$ has been computed in the approximation $T$, and so $T$ is a lower bound for the exact value. Moreover, as two consecutive indices in $\sI$ differ by a factor of $(1+\eps)$, there are at most $(1+\eps)j - j = \eps j$ points $i$ such that $x_j < x_i \leq q$. These are exactly the points whose contribution is not counted in $T$. However, those points are contributing to the goal function by at most $\eps j (q-x_j) \leq \eps T$. Thus $T$ is a $(1+\eps)$ approximation.

\section{Analysis for Sum of Squared Distances}
\label{apx:squaredDist}

In this section, we consider the sum of squared distances for the mis-classified points, and present similar results to the sum of distances case. The proofs are analogous to the sum of distances case and we only include them for the sake of completeness.

\begin{theorem}
There exists a one pass streaming algorithm for the point estimation
variant of the problem under the squared distance in the one dimensional regime, that achieves an
additive error of $\eps$, space of $O(\eps^{-1/3}\sqrt{\log
(1/\eps)})$ words, and that succeeds with constant probability per query. Moreover, $\Omega(\eps^{-1/3})$ is necessary.
\end{theorem}

\begin{proof}

Recall that for $d=1$, the objective simplifies to $F(q) =\tfrac{1}{n}\sum_{i=1}^n (\max\{0,q-x_i\})^2$. We describe a {\em sketching} algorithm that produces a sketch of
size $O(\eps^{-1/3})$ that is able to answer point estimation
queries to this $F$. Later, we show how to adapt this algorithm to the streaming
setting.

Let $m=(\eps^{-1/3})$ and consider two sets of $m$ points. First consider $Y_1,\dots, Y_{m}$ such that $Y_i$ is at position $i/m$. Moreover consider $m$ points $X_1,\dots,X_{m}$ such that $X_i$ is at position $x_{(i\cdot n)/m}$, where we assume that $x_i$'s are in a sorted order, i.e., $x_1 \leq \cdots \leq x_n$. Now sort these $2m$ points and name them $Z_1,\dots, Z_{2m}$. For each $i\leq 2m$ we store three numbers: i) $Z_i$ itself, ii) $s_i$ which is the sum of the distances of the points to the left of $Z_i$ to the point $Z_i$, and  iii) we store $c_i$ which is the number of points $x_i$ to the left of $Z_i$.

Given a query $q\in [0,1]$, we will find $i$ such that $Z_i \leq q
<Z_{i+1}$. We will return $s_i + c_i\cdot (q-Z_i)$. Clearly for the
the points that are to the left of $Z_i$ this distance is computed
correctly. The only points that are not computed in the sum are part
of the points in the interval $[Z_i,Z_{i+1}]$, but we know that there
are at most $n/m = n\eps^{1/3}$ of them (by our choice of the $X_i$'s) and their
distance squared to $q$ is at most $(Z_{i+1} - Z_{i})^2 \leq \eps^{2/3}$ (by our
choice of $Y_i$'s). Therefore they introduce an average error of at most $\eps$ as we require.

The above algorithm can be adapted to the streaming case similar as before. Moreover the same lower bound construction shows that this result is also tight.
\end{proof}

\begin{theorem}
There exists a streaming algorithm for the point estimation variant of
the problem under the squared distance in two dimensional space, that with constant probability,
achieves an additive error of $\eps$, whose sketch size is
$O(\eps^{-4/7})$ words.
\end{theorem}

\begin{proof}
The following shows how to get an $O(1/\eps)$-size sketch
with an additive error of $O(\eps^{7/4})$, and at the end
we just replace $\eps'=\Theta(\eps^{4/7})$ to
get the above lemma.

We use a quad-tree over the unit square $[0,1]\times[0,1]$, where each
node is associated with a number of points (each point is associated
with exactly one node of the quad-tree). Thus each node $v$ contains a
counter $c_v$ for the number of associated points, a randomly chosen
associated point (chosen using reservoir sampling), as well as a sketch
$S_v$ to be described later. Initially, the quad-tree is of depth
$\log (1/\sqrt \eps)$ and all counters/sketches are initialized to
zero. When we stream over a point $p_i$, we associate it with the
corresponding leaf of the quad-tree (process defined later), unless
the counter $c_v$ is already $\eps\cdot n$ and the depth is at least
$2\log 1/\eps$ (i.e., the side length is at least $\eps^2$). In that
case the leaf $v$ is expanded by adding its 4 children, which become
new leaves (with counters initialized to 0).

When we associate a point with a node $v$, we perform two operations. We
increment $c_v$ and update the sketch $S_v$ on the associated
points. The sketch $S_v$ for the associated points, say termed $P_v$,
allows us to compute, for any query line $L=(\theta, b)$, the sum
$\sum_{p\in P_v} (b - \theta^Tp)^2$. The sum $\sum_{p\in P_v} \theta^T p$ can be computed in a
streaming fashion using the sketch from~\cite{pass-glm}. In
particular, the sketch actually consists of two counters: $X_v$, the sum of
the $x$ coordinates, and $Y_v$, the sum of the $y$ coordinates.
Finally, the term $\sum_{p\in P_v}(\theta^Tp)^2$ can be computed using three counters, $X_{vv}$, the sum of the $x$ coordinates squared,  $X_{vv}$, the sum of the $y$ coordinates squared, and  $Z_{xy}$, the sum of the $xy$.


\subparagraph*{Query algorithm.} Given a query line $L$, we distinguish
contribution from points in two types of quad-tree nodes: nodes that do not intersect
the line and those that do. For the first kind, we can just use the
sketch $S_v$ to estimate the distance to the line, without incurring
any error. More precisely we have that 
\begin{align*}
    \sum_{p\in P_v} \dist(p,L)^2 &= \sum_{p\in P_v} (b - \langle p,\theta\rangle )^2\\
&= c_v b^2 - 2b\langle (X_v,Y_v) , \theta \rangle + 2 + \langle (X_{vv},Z_{xy},Y_{vv}),(\theta_x^2, 2\theta_x\theta_y,\theta_y^2)\rangle.
\end{align*}
Note that this is
included in the final sum iff the entire node lies in the halfplane $\langle x,\theta\rangle\leq b$.

For the second kind of nodes, we estimate their contribution as
follows. For each non-empty node $v$, with the random sample $r_v$, we add to
the final sum the quantity $\tfrac{1}{n}\cdot c_v\cdot
(\max\{0,b-\theta^Tr_v\})^2$. 


\subparagraph*{Sketch size.} It is clear that the total space usage of the
algorithm is at most order of the size of the quad-tree. The size of
the tree is bounded by $O(1/\eps)$ as follows. First, there at most
$O(1/\eps)$ of the original nodes. Second, each new children created
has the property that its parent got associated with $\eps n$ points,
hence at most $4/\eps$ such children can be ever created.


\subparagraph*{Error analysis.} We now analyze the error of the sketching
algorithm. For the nodes that do not cross the line $L$, the
distances of their points are computed exactly. So we only need to
argue about the crossing nodes.  First of all, note that we can ignore
all leaves with more than $\eps n$ points at them as their diameter is
less than $2\eps^2$. 

Let $\mathcal{C}$ be the set of leaves that cross the query line and
have diameter at least $2\eps^2$ (and hence less than $\eps n$
associated points). It is immediate to check that, in expectation, our
estimator outputs the correct value; in particular for $P_v$ the set
of points associated to a node $v$:
\begin{align*}
\E\left[\sum_{v\in \mathcal{C}} c_v (\max\{0,\dist(r_v,L)\})^2\right] &= \sum_{v\in \mathcal{C}}\sum_{p\in P_v} \frac{1}{c_v} c_v (\max\{0,\dist(p,L)\})^2 \\
&= \sum_{v\in \mathcal{C}}\sum_{p\in P_v} (\max\{0,\dist(p,L)\})^2. 
\end{align*}
Thus we only need to argue that it concentrates closely around its
expectations, with constant probability. Let us compute the
variance. The point in each (non-empty) node is chosen independently at
random. Thus we can sum up the variances of each node. Consider a
node $v$ with side length $\ell\ge \eps^2$. Then we have that
\[
 \Var\left[ c_v (\max\{0,\dist(r_v,L)\})^2\right]  \leq \sum_{i=1}^{c_v} \frac{1}{c_v}\cdot (c_v\cdot\ell^2)^2 \leq c_v^2\ell^4 \leq (n\eps)^2\ell^4.
\]
Now note that any line can intersect only $(1/\ell)$ nodes with side
length $\ell$, and thus the total variance of all nodes with side
length $\ell$ is at most $(n\eps)^2\ell^3$. Hence, the total variance
over all nodes (over all levels) is at most
$O(n^2\eps^2\eps^{3/2})$. Overall, the standard deviation is at most $
O(n\eps^{7/4})$. By Chebyshev's bound, the reported answer has an
additive error of $O(n\eps^{7/4})$ with constant probability.

As stated earlier, replacing $\eps' = \Theta(\eps^{4/7})$ we get that the
algorithm is providing an additive $\eps'$ approximation using space
$\tilde O((\eps')^{-4/7})$, completing the proof of the result.
\end{proof}


\end{document}